\documentclass[12pt]{article}
\usepackage{amsmath, amsthm, amssymb}
\usepackage{graphicx,psfrag,epsf}
\usepackage{enumerate}
\usepackage{natbib}
\usepackage{xcolor}
\usepackage{url}
\usepackage{algorithm, algpseudocode}

\def\W{X}

\def\QQ{\mathbb{Q}}
\def\QQp{\mathbb{Q}^+}

\def\WW{\mathbb{W}}
\def\ZZ{\mathbb{Z}}
\def\ZZp{\mathbb{Z}^+}

\def\df{\mathrm{d}}
\def\Rl{\mathbb{R}}

\def\normal{N}
\def\prior{\pi}
\def\brobri{\text{BB}}
\def\fearn{\texttt{FearnpMCMC}}
\def\eul{\texttt{EulpMCMC}}
\def\pf{\texttt{FearnPF}}

\newtheorem{theorem}{Theorem}[section]

\newtheorem{proposition}[theorem]{Proposition}
\newtheorem{corollary}[theorem]{Corollary}

\newtheorem{definition}{Definition}

\begin{document}

\def\spacingset#1{\renewcommand{\baselinestretch}%
{#1}\small\normalsize} \spacingset{1}


{
 \title{\bf An Exact Auxiliary Variable Gibbs Sampler for a Class of Diffusions}
 \author{Qi Wang 
   \hspace{.2cm}\\
  Department of Statistics, Purdue University\\
  and \\
  Vinayak Rao\\
  Department of Statistics, Purdue University\\
  and \\
  Yee Whye Teh \\
  Department of Statistics, University of Oxford}
  \maketitle
}

\bigskip

\begin{abstract}
  Stochastic differential equations (SDEs) or diffusions are continuous-valued continuous-time stochastic processes widely used in the applied and mathematical sciences. 
  Simulating paths from these processes is usually an intractable problem, and typically involves time-discretization approximations. 
We propose an exact Markov chain Monte Carlo sampling algorithm that involves no such time-discretization error. 
Our sampler is applicable to the problem of prior simulation from an SDE, posterior simulation conditioned on noisy observations, as well as parameter inference given noisy observations. 
Our work recasts an existing rejection sampling algorithm for a class of diffusions as a latent variable model, and then derives an auxiliary variable Gibbs sampling algorithm that targets the associated joint distribution. 
At a high level, the resulting algorithm involves two steps: simulating a random grid of times from an inhomogeneous Poisson process, and updating the SDE trajectory conditioned on this grid.
Our work allows the vast literature of Monte Carlo sampling algorithms from the Gaussian process literature to be brought to bear to applications involving diffusions. 
We study our method on synthetic and real datasets, where we demonstrate superior performance over competing methods.
\end{abstract}

\noindent%

\spacingset{1.45}

\section{Introduction}
\label{sec:intro}
Diffusion processes are a class of stochastic processes that have been deeply studied and widely applied across a variety of theoretical and applied domains. 
A diffusion evolves though time according to a stochastic differential equation (SDE)~\citep{oksendalstochastic}, and is a continuous-time Markov process whose realizations are continuous paths. 
The most well-known example is Brownian motion, corresponding to a random walk through some  finite-dimensional Euclidean space. 
Brownian motion is characterized by two fixed parameters: a drift coefficient $\alpha$, and a diffusion coefficient $\sigma$. 
SDEs generalize this, allowing the drift and diffusion to depend on the current state of the process. 
A simple example is the Ornstein-Uhlenbeck (OU) process~\citep{uhlenbeck1930theory}, where the drift equals the negative of the difference between the current state and some constant $\mu$, resulting in mean-reverting dynamics. 
Closely related is the Brownian bridge~\citep{oksendalstochastic}, where the drift at any time $t$ is this negative difference divided by $T-t$, the time remaining till the end of an interval $[0,T]$. 
This ensures that with probability one, the process ends at $\mu$ at time $T$. 
The OU process and the Brownian bridge are still simple Gauss-Markov processes, with the distribution over the process value at some future time following an easy-to-compute normal distribution. 
More general drift and diffusion dependencies allow SDEs to model rich, mechanistic, nonlinear and nonstationary phenomena from a variety of applied disciplines.
Examples include astronomy~\citep{schuecker2001cosmic}, biology~\citep{ricciardi2013diffusion}, psychology~\citep{tuerlinckx2001comparison}, ecology~\citep{holmes2004beyond}, economics~\citep{bergstrom1990continuous}, 
genetics~\citep{lange2003mathematical}, 
finance~\citep{black1973pricing},  physics~\citep{papanicolaou1995diffusion}, and 
political and social sciences~\citep{cobb1981stochastic}. 

The flexibility that SDEs offer comes at a severe computational cost, especially in data-driven applications. 
With a few exceptions, the nonlinear, continuous-time dynamics of SDEs result in distributions over future values that are not just non-Gaussian, but also unavailable in closed form. 
If an SDE forms a prior distribution over paths, then even simulating from this distribution forms an intractable problem. 
Given noisy measurements via some measurement process, posterior simulation is an even more challenging problem. 
As a consequence, both prior and posterior simulation are typically carried out by imposing approximations through time-discretization, common examples being Euler-Maruyama or Millstein approximations~\citep{kloeden2012numerical}. 
While this allows ideas from the discrete-time literature to be used, time-discretization introduces errors into inferences, and controlling these requires fine discretization grids and expensive computation. 

Our main contribution in this paper is an auxiliary variable Markov chain Monte Carlo (MCMC) algorithm that targets the posterior distribution over paths {\em exactly} without any such approximations. 
Our scheme builds on a rejection sampling algorithm that allows exact simulation from a class of SDEs, outlined in the papers~\citet{beskos2005exact, beskos2006retrospective, beskos2006exact}. 
Our work recasts this rejection sampling algorithm as a latent variable model, and then derives a Gibbs sampling algorithm that at a high level involves two simulation steps: 1) simulate a random grid of times from an inhomogeneous Poisson process conditioned on a set of diffusion values, and 2) update the diffusion values on this Poisson grid.
Our algorithm allows us to easily use standard tools from the vast Gaussian process literature~\citep{titsias2008markov}, and also allows conditional simulation given noisy observations. 
Our focus is mostly on one-dimensional diffusions, although our ideas also apply to some multivariate diffusions which can be transformed to have a constant diffusion function $\sigma(\cdot)$. 
A more serious restriction is that like~\citet{beskos2005exact}, our algorithm applies to diffusions whose Radon-Nikodym derivative with respect to a biased Brownian bridge is bounded (see section~\ref{sec:ea1}): we call these SDEs of class EA1.
It is conceptually easy to see how our basic idea extends to larger classes of diffusions (called EA2 and EA3), though these generalizations can be quite involved. We leave this for future work.

We organize our paper as follows. Section~\ref{sec:sdes} briefly introduces stochastic differential equations and describes the exact EA1 algorithm of~\cite{beskos2005exact}. 
Section~\ref{sec:posterior} sets up the general Bayesian model for which we wish to carry out posterior inference, and describes our proposed MCMC algorithm in this broader setting. 
Section~\ref{sec:param} shows how to extend our MCMC algorithm to incorporate parameter sampling.
We discuss related work in section~\ref{sec:related}, while in section~\ref{sec:exp} and~\ref{sec:example}, we evaluate our, and three other sampling algorithms, on synthetic and real datasets.

\vspace{-.2in}
\section{Stochastic differential equations (SDEs)}
\label{sec:sdes}
\vspace{-.1in}
A diffusion $\W_t$ is a continuous-valued continuous-time Markov process that solves the SDE
\vspace{-.1in}
\begin{align} \label{eq:sde_general}
\df \W_t &= \alpha_\theta(\W_t) \df t + \sigma_\theta(\W_t) \df B_t . 
\vspace{-.2in}
\end{align}
The process is driven by a Brownian motion whose value at time $t$ is $B_t$. 
The functions $\alpha_\theta(\cdot)$ and $\sigma_\theta(\cdot)$ are the {\em drift} and {\em diffusion} terms respectively, while $\theta$ represents parameters governing the system dynamics. 
For clarity, we drop dependencies on $\theta$ until section~\ref{sec:param} on parameter sampling.

Informally, equation~\eqref{eq:sde_general} implies that $\df \W_t$, the infinitesimal change in the value of the diffusion at time $t$, is comprised of two parts, a deterministic and a stochastic component. 
The former is determined by the current value $\W_t$ of the diffusion transformed by $\alpha(\cdot)$, while the latter is an increment of Brownian motion $\df B_t$ scaled by $\sigma(\W_t)$. 
In general $\W_t$ and $B_t$ can be $d$-dimensional vectors, with $\alpha(\W_t) \in \Rl^d$ and $\sigma(\W_t) \in \Rl^{d \times d}$. 
For one-dimensional diffusions, all these are scalars. 

In this paper, as in~\citet{beskos2005exact} and follow-up papers, we will assume that the diffusion coefficient $\sigma(\cdot)=1$. 
Thus, we will be dealing with diffusions solving the equation
\begin{align} \label{eq:sde}
  \df \W_t = \alpha(\W_t)\df t + \df B_t.
\end{align}
For one-dimensional diffusions, this is a mild assumption, since a general SDE can be transformed to have a diffusion coefficient of one via the Lamperti transform~\citep{moller2010state}. 
This involves scaling the diffusion by the function 
$ \eta(x) = \int_{-\infty}^{x} \frac{1}{\sigma(u)} \df u$. 
Now, the process  $\W'_t = \eta(\W_t)$ is a diffusion with 
diffusion coefficient equal to 1~\citep{moller2010state}. 
In higher-dimensions, the restriction to constant diffusion is more significant, since the Lamperti transform typically does not exist, and since there 
needs to exist a function $A(\textbf{u})$ satisfying $\alpha(\textbf{u}) = \nabla A(\textbf{u})$~\citep{fearnhead2010random}. 
In what follows, we assume such a transformation has been applied to produce our SDE of interest. 

\vspace{-.2in}
\subsection{Simulation via the Euler-Maruyama Method} 
The Euler-Maruyama method \citep{iacus2009simulation, kloeden2012numerical} forms the simplest approach to simulating general diffusions over an interval $[0,T]$. 
This simplicity comes at the price of approximation error. 
Under the Euler-Maruyama scheme, one chooses a time-discretization granularity $\Delta t$, with the change in the diffusion value $ \Delta \W_t := \W_{t + \Delta t} - \W_t$ approximated as 
\vspace{-.1in}
\begin{align} \label{eq:euler}
\Delta \W_{t} \approx \alpha(\W_t ) \Delta t + \sigma(\W_t) \Delta B_t,
\end{align}
where $\Delta B_t \sim N(0, \Delta t)$. Effectively, the change $\Delta \W_t$ follows a conditionally Gaussian distribution:
\begin{align} \label{eq:euler2}
  \Delta \W_{t} \sim N(\alpha(\W_t) \Delta t, \sigma(\W_t) \Delta t).
\end{align}
The discretization error from the Euler-Maruyama method can be reduced by using a finer discretization grid. 
Alternately, more sophisticated approaches like the Millstein algorithm~\citep{kloeden2012numerical} can provide more accurate approximations for a fixed time resolution. 


\subsection{An exact simulation algorithm (EA1) for diffusion processes} \label{sec:ea1}
Algorithms like the Euler-Maruyama method allow easy path simulation from general SDEs at the price of discretization error. 
The algorithm of~\citet{beskos2005exact} on the other hand allows {\em exact} simulation from a subclass of SDEs with diffusion coefficient 1. 
In follow-up work~\citep{beskos2006retrospective}, this was extended to a broader class of such SDEs, though we focus on the original algorithm, called the 
Exact algorithm 1 or EA1. 
We refer to the associated family of SDEs as class EA1, which we characterize below. 
At a high level, EA1 is a rejection sampling scheme, where proposals are made from a simple stochastic process (Brownian motion), and are accepted or rejected with appropriate probability. 
The ingenuity of the algorithm lies in a retrospective sampling scheme that only requires evaluating the paths on a finite set of times. 

Assume that at time $0$ the diffusion has initial value $X_0 = x$; later we will place a probability $\prior$ over $\W_0$. 
Since the diffusion coefficient equals 1, the resulting stochastic process differs from standard Brownian motion only through the drift function $\alpha(\cdot)$. 
Informally, this results in paths from the SDE having the same `roughness' as the Brownian motion paths. 
A consequence is that the probability measure over paths specified by the diffusion process is absolutely continuous with respect to the probability measure corresponding to Brownian motion. 
This is formalized by {Girsanov's theorem }\citep{oksendalstochastic}, that characterizes the diffusion process via a Radon-Nikodym derivative with respect to Brownian motion. 

Write $\mathcal{C}$ for the space of continuous functions on $[0,T]$. We will refer to generic elements of this space as $\omega$. For paths $\omega$ with initial value $x$, write $\WW_x$ and $\QQ_x$ for path probability measures corresponding to Brownian motion and the SDE respectively. 
Then, under standard assumptions (we refer to~\citet{oksendalstochastic} for more details), we have
\begin{theorem}[Girsanov's theorem]
The Radon-Nikodym derivative $\frac{\df \QQ_x}{\df \WW_x}$ satisfies
\begin{align}
  \frac{\df \QQ_x}{\df \WW_x}(\omega) =
  \exp\left\{  \int_0^T \alpha(\omega_t) \df \omega_t - \frac{1}{2}\int_0^T \alpha^2(\omega_t)  \df t \right\}.
  \label{eq:girsanov}
\end{align}
\end{theorem}
Let $A(u) = \int_0^u \alpha(t) \df t$, and recall the definition of a Brownian bridge: this is just a Brownian motion conditioned on its end points. 
For a density $h_x(u) \propto \widetilde{h}_x(u) := \exp(A(u) - (u - x)^2/2T)$, define an $h_x$-biased Brownian bridge as a stochastic process starting at $x$, ending with a value $\W_T$ drawn from $h_x$, with the two points linked by a Brownian bridge. 
Write $\ZZ_x$ for the law of this process. 
Note that for this to be well defined, $\widetilde{h}_x$ must be normalizable, so that the integral $\int \widetilde{h}_x(u)\df u = \int \exp(A(u) - (u - x)^2/2T) \df u$ is finite.
Then we have~\citep{beskos2005exact}:
\begin{proposition}
  Let the drift function $\alpha$ satisy the conditions of Girsanov's theorem and be continuously differentiable. Then
\begin{align}
  \frac{\df \QQ_x}{\df \ZZ_x}(\omega) &\propto
  \exp\left\{ -\frac{1}{2}
             \int_0^T \left(\alpha^2(\omega_t) + \alpha'(\omega_t)  \right) \df t \right\}.
             \label{eq:sde_poiss_orig}
\end{align}
\label{prop:girs}
\end{proposition}
\vspace{-.5in}
\begin{proof}
    Write $A_t = A(\omega_t)$. By It\^{o}'s lemma~\citep{oksendalstochastic},
\begin{align}
  \df A_t &= \frac{\partial A_t}{\partial t} \df t + \frac{\partial A_t}{\partial \omega_t} \df \omega_t + \frac{1}{2} \frac{\partial^2 A_t}{\partial \omega_t^2} \df t  
       = 0 + \alpha(\omega_t) \df \omega_t + \frac{1}{2} \alpha'(\omega_t) \df t.
\end{align}
Solving for $\int_0^T \alpha(\omega_t) \df \omega_t $ and substituting in equation \eqref{eq:girsanov}, we get
\begin{align}
  \frac{\df \QQ_x}{\df \WW_x}(\omega) =
  \exp\left\{ A(\omega_T) - A(\omega_0) -\frac{1}{2}
  \int_0^T \left(\alpha^2(\omega_t) + \alpha'(\omega_t) \right) \df t \right\}. \label{eq:joint_unnorm}
\end{align}
By definition, the measure $\ZZ_x$ is a reweighting of $\WW_x$ by $h_x(\omega_T)$. Thus, 
\begin{align}
  \frac{\df \QQ_x}{\df \ZZ_x}(\omega) &\propto
  \exp\left\{ -A(\omega_0) -\frac{1}{2}
             \int_0^T \left(\alpha^2(\omega_t) + \alpha'(\omega_t)  \right) \df t \right\}. \label{eq:sde_poiss1}
\end{align}
Since we are fixing $\omega_0=x, A(\omega_0)$ is a constant, and the result follows.
\end{proof}
We now come to the key assumption of the EA1 algorithm of~\citet{beskos2005exact}:
\begin{definition}
  \label{def:ea1}
  An SDE belongs to class EA1 if it satisfies the assumptions of Proposition~\ref{prop:girs}, and its drift function $\alpha$ satisfies $\frac{1}{2}\left(\alpha^2(\cdot) + \alpha'(\cdot) \right) \in [L, L+M]$ for finite $L$ and $M$.
\end{definition}
\noindent We focus on SDEs of class EA1. Adding and subtracting $L$ from the exponent in equation~\eqref{eq:sde_poiss1}, 
\begin{align}
  \frac{\df \QQ_x}{\df \ZZ_x}(\omega) &\propto
  \exp\left\{ -\frac{1}{2}
             \int_0^T \left(\alpha^2(\omega_t) + \alpha'(\omega_t) - 2L \right) \df t \right\} 
                                     := \exp\left\{ - \int_0^T \phi(\omega_t)  \df t \right\} 
                                    := \rho(\omega). 
             \label{eq:sde_poiss}
\end{align}
For class EA1, the function $\phi(\cdot) = \frac{1}{2}(\alpha^2(\cdot) + \alpha'(\cdot)-2L)$ is positive, and exponentiating its negative integral gives a number $\rho(\omega)$ between $0$ and $1$. 
This suggests a rejection sampling scheme~\citep{Robert2005} to simulate from $\QQ_x$: propose a path from the stochastic process $\ZZ_x$, and accept it with probability $\rho(\omega)$. 
Naively, this requires 1) simulating the entire path $\omega$, and 2) transforming and integrating $\omega$ to calculate $\rho(\omega)$, both being impossible steps. 
The EA1 algorithm bypasses this by recognizing that equation \eqref{eq:sde_poiss} gives the probability that a Poisson process with intensity $\{\phi(\omega_t), t\in [0,T]\}$ produces $0$ events on the interval $[0,T]$. 
It takes the approach of partially `uncovering' the path $\omega$, simulating it on a finite set of times, until the number of Poisson events is determined. 
To do this, the EA1 algorithm exploits the bound $\phi(\cdot) \le M$ to simulate a rate-$\phi(\omega)$ Poisson process via the thinning theorem~\citep{Lewis1979}. 
It does this in three steps: simulate a Poisson process $\Psi$ with intensity $M$ on the interval $[0,T]$, instantiate an $h$-biased Brownian bridge $\omega_t$ on $\Psi$, and keep each point $t_i \in \Psi$ with probability $\phi(\omega_{t_i})/M$. 
The surviving points then form an {\em exact} realization from a rate-$\phi(\omega)$ Poisson process. 
The probability this Poisson process has $0$ events is given by equation~\eqref{eq:sde_poiss}.

Now, the EA1 algorithm involves repeatedly simulating from the rate-$\phi(\omega)$ Poisson process this way until a realization with no events is produced. 
We write the corresponding path as $\W$, this forms an exact realization of the SDE of interest. 
Note that at this stage, we only have $\W_0, \W_T$ and $\W_\Psi$, the last being the values of the diffusion uncovered on the times in the Poisson set $\Psi$. 
We will refer to the pair $(0 \cup \Psi \cup T, \W_0 \cup \W_\Psi \cup \W_T)$ as the diffusion `skeleton', 
this forms a sufficient statistic that allows the diffusion at any other set of times to be easily and exactly simulated. 
For this, we recognize that the accepted path was a proposal from a biased Brownian bridge, but which only was evaluated at times in $0 \cup \Psi \cup T$. 
It can retrospectively be uncovered at a set of times by conditionally simulating from a Brownian bridge. 
Consider a set of times $G$ between two successive elements $t_i$ and $t_{i+1}$ of $\Psi$. 
We simulate $\W_G$, the diffusion evaluated on $G$, from a Brownian bridge with endpoints $\W_{t_i}$ and $\W_{t_{t+1}}$. 
We write this as 
$X_G \sim \brobri_G(t_i,\W_{t_i},t_{i+1},\W_{t_{i+1}})$ (see equation~\eqref{eq:brownian_br} in the appendix).
Algorithm~\ref{alg:EA1} describes all steps involved with the EA1 algorithm. 
\begin{algorithm}[h]
  \caption{Simulate an SDE of class EA1 with drift term $\alpha(\cdot)$ over an interval $[0,T]$}
   \label{alg:EA1}
  \begin{tabular}{l l}
    \textbf{Input:  } & \text{An initial distribution over the diffusion state $\prior(\cdot)$},
                       \text{a finite grid of times  $G \in [0,T]$.} \\
    \textbf{Output:  }& \text{A diffusion skeleton $(0 \cup \Psi \cup T, \W_0 \cup \W_\Psi \cup \W_T)$.} \\
     & The diffusion values \text{$\W_G$ evaluated on the grid $G. \qquad \qquad \qquad \qquad \qquad \qquad \qquad $} \\
   \hline
   \end{tabular}
   \begin{algorithmic}[1]
     \State {Calculate $A(\cdot), \phi(\cdot)$, and the constants $L$ and $M$ from $\alpha(\cdot)$,} and
     set {\texttt{accept}} to {\texttt{false}}.
     \While{\texttt{accept = false}}  \Comment{Rejection sampling}
     \State {Simulate a rate-$M$ Poisson process $\Psi = \{t_1, t_2, \cdots, t_{|\Psi|}\}$  on $[0,T]$.}
      \State At the start time $0$, simulate the initial value $\W_0$ of the diffusion from $\prior$.
      \State At the end point $T$, simulate $\W_T$ from $h_{\W_0}(\W_T) \propto \exp(A(\W_T) - (\W_T - \W_0)^2/2T)$.
      \State Simulate a Brownian bridge connecting $(0, \W_0)$ and $(T, \W_T)$ on the times $\Psi$:
            \vspace{-.1in}
          \begin{align}
            \label{eq:brownianbridge}
            X_\Psi \sim \brobri_\Psi(0,\W_{0},T,\W_{T}) \qquad \text{(see equation~\eqref{eq:brownian_br} for details)}.
          \end{align}
            \vspace{-.4in}
      \State {For $i \in \{1,\dotsc, |\Psi|\}$, simulate $u_i \sim$ Uniform$(0,1)$. 
    If all $u_i > \frac{\phi(\W_{t_i})}{M}$, set {\texttt{accept = true}}. 
  } 
      \EndWhile
      \For {i in $\{0,\dotsc,|\Psi|\}$} \Comment{Impute diffusion on $G$} 
      \State 
      Define $t_0=0, t_{|\Psi|+1}=T$ and $G_{i} = G \cap (t_i,t_{i+1})$.
      Simulate $\W_{G_{i}}\sim \brobri_{G_{i}}(t_i,\W_{t_i},t_{i+1},\W_{t_{i+1}})$.
      \EndFor
   \end{algorithmic}
\end{algorithm}

\vspace{-.1in}
\section{Posterior simulation for SDEs} \label{sec:posterior}
\vspace{-.1in}
The EA1 algorithm, while exact, can suffer from high rejection rates. 
This happens when dealing with long time intervals, or when the drift $\alpha(\cdot)$ causes $\QQ_x$ to differ significantly from the 
proposal $\ZZ_x$. 
Further, the EA1 algorithm is primarily designed to simulate from an SDE prior or an end-point conditioned SDE. 
As we describe below, extending it to simulating SDE paths from conditional distributions given noisy observations can be challenging. 
Our proposed sampler aims to address both these problems, and brings sampling algorithms from the Gaussian process literature to applications with SDEs. 
Before describing our algorithm, we set up the general problem. 

\vspace{-.08in}
\subsection{Bayesian model} \label{sec:bayes_model} 
\vspace{-.1in}
Consider a latent trajectory $X = \{X_t: t \in [0,T]\}$ on the interval $[0,T]$. 
We model this as a realization of an SDE of class EA1, with drift $\alpha(\cdot)$, and distribution $\prior$ on the initial state $X_0$. 
Following our previous notation, our prior distribution on the process $X_0 \times \{X_t: t \in (0,T]\}$ equals the product measure $\prior \times \QQ_{X_0}$. 
Write this as $\QQ_{\prior}$. 
We are given noisy measurements of the latent trajectory, with likelihood $\ell(X)$. 
We will assume this depends only on the trajectory values at a finite set of times $O=\{o_1,\dotsc,o_{|O|}\}$ (without loss of generality, we let $O$ include $0$ and $T$), so that $\ell(X) = \ell(X_O)$ (recall that $\W_O$ is the diffusion evaluated on the times in $O$). 
A simple example is when we have i.i.d.\ additive-noise measurements $Y_O = \{y_1,\cdots,y_{|O|}\}$ at the times $O$, so that $\ell(X) = \prod_{o \in O} \ell_o(X_o)$ and for example, $\ell(X_o) = N(y_o|X_o,\sigma_Y^2)$. 
We can also consider more complex likelihoods, where this condition holds after augmenting the observations with additional variables. 
Examples of such likelihoods include point processes~\citep{adams-murray-mackay-2009b, RaoTeh2011b}, jump processes~\citep{RaoTeh13}, or even other diffusions modulated by the latent SDE trajectory. 
In this case, our MCMC sampler will include such data-augmentation as an inner step. 
Obviously, our setup includes the problem of prior simulation, where there are no observations.

Our goal is to simulate from the conditional distribution over paths under a prior $\QQ_{\pi}$, given the observations with likelihood $\ell(\W)$. Write this as $\QQ_{\pi,\ell}$, which forms the posterior distribution over paths under our Bayesian model. Observe that this satisfies 
\begin{align} 
 \frac{\df \QQ_{\pi,\ell}}{\df \QQ_\prior}(\omega) \propto \ell(\omega). 
 \label{eq:sde_post}
\end{align} 
The EA1 algorithm, as outlined in section~\ref{sec:ea1}, can only simulate trajectories from the prior distribution $\QQ_\prior$. 
In~\citet[][Section 6.2]{beskos2006retrospective}, the authors adapt the EA1 rejection-sampling algorithm to conditional simulation where the diffusion is exactly observed at a finite set of times. 
One can adapt this when the diffusion is noisily observed, repeating two steps:
1) given $\W_O$, the diffusion imputed on the observation times $O$, use the rejection-sampling algorithm to simulate an SDE skeleton within each sub-interval $(o_i, o_{i+1})$, and 2) conditioned on the skeleton, update the diffusion values at the observation times $O$. 
The first step exploits the Markov property of the SDE, and runs the conditional EA1 algorithm independently for each interval $[o_i,o_{i+1}]$. 
The second step involves simulating each $X_o, o \in O$ from the conditional distribution resulting from a Brownian bridge prior on $\W_o$ (algorithm~\ref{alg:EA1}, line 10) and the likelihood $\ell(\W_o)$. 

Such an approach, while useful, can scale badly with high observation-rates, as is common in fields like high-frequency finance. 
Even with low to moderate observation rates, it can be necessary to partition the observation interval into small sub-intervals to maintain low rejection rates~\citep[][Section 4]{beskos2006retrospective}. 
This slows down MCMC mixing, since 1) we are instantiating more of the diffusion path, and 2) rather than updating the entire path in a single step, we conditionally update part of the trajectory given the rest. 
The finer the sub-intervals, the stronger the coupling, and thus, the poorer the mixing. 
Our proposed algorithm eliminates the rejection sampling step altogether, instead allowing practitioners to exploit MCMC algorithms for Gaussian process models in a fairly straightforward fashion.

\vspace{-.1in}
\subsection{Our proposed auxiliary variable Gibbs sampler for SDEs} \label{sec:aux_gibbs}
We describe an MCMC sampling algorithm that targets the 
posterior distribution over trajectories $\QQ_{\prior,\ell}$ from our Bayesian model of the previous section. 
Note that this equals the prior $\QQ_{\prior}$ when the likelihood $\ell(\cdot)$ is a constant function. 
To keep our notation simple, we will write $\QQ_{\pi,\ell}$ as $\QQ$. 
%
Recall that $\ZZ_x$ is an $h$-biased Brownian bridge starting at $x$.
In a similar manner to $\QQ_x$, use $\ZZ_x$ to define $\ZZ_\prior$ and $\ZZ_{\prior,\ell}$. 
Thus, $\ZZ_\prior$ is the distribution over Brownian bridge paths, with values at time $0$ and $T$ distributed as $\prior$ and $h_{\W_0}$ respectively. 
Treating this as a prior over paths, $\ZZ_{\prior,\ell}$ is the posterior distribution corresponding to observations with likelihood $\ell(\omega)$. 
Again, we set $\ZZ$ as equal to $\ZZ_{\prior,\ell}$. 
It follows directly from equations~\eqref{eq:sde_poiss} and~\eqref{eq:sde_post} that for a path $\omega \in \mathcal{C}$, 
\begin{align}
  \frac{\df \QQ}{\df \ZZ}(\omega) \propto \exp\left\{-\int_0^T \phi(\omega_t)\df t\right\}. \label{eq:sde_poiss2}
\end{align}
Write $\mathcal{M}$ for the space of finite point process realizations on the interval $[0,T]$. Let $\mathbb{M}$ be the probability measure on $\mathcal{M}$ corresponding to a rate-$1$ Poisson process. 
Define the product measure $\ZZp = \ZZ \times \mathbb{M}$. 
For $\Psi \in \mathcal{M}$, and recalling that $M$ is the supremum of $\phi(\cdot)$, define the measure $\QQp$ via the following Radon-Nikodym derivative with respect to $\ZZp$: 
\begin{align}
  \frac{\df \QQp}{\df \ZZp}(\omega,\Psi) =  
  \exp(-MT) \prod_{t \in \Psi} \left(M - {\phi(\omega_t)}\right).
  \label{eq:joint_data}
\end{align}
\begin{proposition}
  $\QQ^+$ has $\QQ$ as its marginal distribution:
  $\int_{\mathcal{M}}  \mathrm{d}\QQp(\omega,\Psi) =  \mathrm{d}\QQ(\omega)$.
  \label{prop:target_marg}
\end{proposition}
\begin{proof}
  From equation~\eqref{eq:joint_data}, we have 
  $$\int_{\mathcal{M}} \mathrm{d}\QQp(\omega,\Psi)  = 
  \int_{\mathcal{M}}  \mathrm{d} \ZZp(\omega,\Psi) \exp(-MT)\prod_{t \in \Psi} \left(M- {\phi(\omega_t)} \right) =
  \exp(-MT) \mathrm{d} \ZZ(\omega) \mathbb{E}_{\mathbb{M}} \left[\prod_{t \in \Psi} \left(M- {\phi(\omega_t)} \right)\right],
$$
where $\mathbb{E}_\mathbb{M}$ is the expectation with respect to the Poisson measure $\mathbb{M}$. 
By Campbell's theorem \citep{kingman1992poisson}, we have
$\mathbb{E}_{\mathbb{M}}\left[\prod_{t \in \Psi} \left(M- {\phi(\omega_t)}\right)\right] = 
\mathbb{E}_{\mathbb{M}}\left[\exp\left\{\sum_{t \in \Psi} \log\left(M- {\phi(\omega_t)}\right)\right\}\right] = $
$\exp\left\{\int_0^T (M-\phi(\omega_t)) \df t \right\}$. 
The result 
follows directly from this and equation~\eqref{eq:sde_poiss2}. 
\end{proof}
While our goal is to produce samples from $\QQ$, our MCMC sampler 
is an auxiliary variable sampler that targets the joint distribution $\QQ^+$. 
Its state-space is the SDE trajectory $\W$ as well as the random set of Poisson times $\Psi$. 
Proposition~\ref{prop:target_marg} tells us that discarding the Poisson times $\Psi$ produces trajectories $\W$ from the desired conditional distribution $\QQ$. 
Our algorithm takes a Gibbs sampling approach, and targets the distribution $\QQp$ by repeating two steps: simulate Poisson times $\Psi$ given the path $\W$, and update the path given the Poisson times. 
Equation~\eqref{eq:joint_data} allows us to derive two simple corollaries that underpin our Gibbs sampler. 
\begin{corollary}
  Conditioned on the trajectory $\W$, the point events $\Psi$ follow an inhomogeneous Poisson process with rate $(M-\phi(\W))$.
\label{corr:psi_cond}
\end{corollary}
\begin{proof}
  For $\omega$ fixed to $\W$, from equation~\eqref{eq:joint_data}, $\Psi$ is a point process whose density 
  with respect to $\mathbb{M}$ is proportional to
  $\exp(-MT)\prod_{t \in \Psi} \left(M - {\phi(\W_t)}\right)$.
  Write this as $\mathbb{M}_\W$.
  For any nonnegative function $g$, the Laplace functional $\mathbb{E}_{\mathbb{M}_\W}[\exp\left\{-\sum_{t \in \Psi}g(t)\right] \propto 
  \mathbb{E}_{\mathbb{M}}\left[\exp\left\{-MT-\sum_{t \in \Psi} g(t)\right\}\prod_{t \in \Psi} \left(M- {\phi(\W_t)}\right)\right]$.
  From Campbell's theorem, this equals  
  $\exp\left\{-MT+\int ((M-\phi(\W_t))e^{-g(t)}\mathrm{d}t \right\}$, 
  which is proportional to the Laplace functional of a rate-$(M-\phi(X))$ Poisson process~\citep{kingman1992poisson}.
\end{proof}
An important point to note is that conditioned on $\W$, the distribution over $\Psi$ does not depend on the likelihood $\ell(\W)$, 
since the observations depend only on the path values $\W$. 
Instead, they enter when we update $\W$.
Our second corollary concerns updating $\W$ given the Poisson times $\Psi$. 
\begin{corollary}
  Conditioned on the Poisson times $\Psi$, the trajectory $\W$ has density with respect to $\ZZ_{\prior}$ 
  given by $h_{\W_0}(\W_T)\ell(\W_O) \prod_{t \in \Psi} \left( 1- \frac{\phi(\W_t)}{M}\right)$.
\label{corr:omega_cond}
\end{corollary}
\begin{proof}
  Conditioned on the times $\Psi$, from equation~\eqref{eq:joint_data}, we see that $\W$ has density with respect to $\ZZ$ proportional to $\prod_{t \in \Psi} \left(1- \frac{\phi(\W_t)}{M}\right)$. 
  The result follows from the definition of $\ZZ=\ZZ_{\pi,\ell}$.
\end{proof}
The above result shows us that conditioned on the Poisson skeleton $\Psi$, the probability density of the SDE path evaluated $\Psi$ and $O$ (write this as $\W_{\Psi \cup O}$) is given by 
\begin{align}
  p(\W_{\Psi \cup O}) &\propto \prior(\W_0)h_{\W_0}(\W_T) \brobri(X_{O \cup \Psi}|0,X_0,T,X_T) \ell(\W_O) \prod_{g \in \Psi} \left( 1- \frac{\phi(\W_g)}{M}\right). 
  \label{eq:hmc_target}
\end{align}
This corresponds to a fairly typical posterior distribution in applications involving Gaussian processes~\citep{williams2006gaussian}. Here our prior over trajectories is the Brownian motion prior $\WW_\pi$, and our likelihood is $h_{\W_0}(\W_T) \ell(\W_O) \prod_{g \in \Psi} \left( 1- \frac{\phi(\W_g)}{M}\right)$.
Consequently, after conditioning on the Poisson grid $\Psi$, we do not need to calculate intractable SDE transition probabilities to calculate prior probabilities over the trajectory $\W$. 
The SDE posterior is amenable to standard Gaussian process MCMC techniques. 
For a survey of such methods, see for example~\citet{titsias2008markov}, we will use Hamiltonian Monte Carlo~\citep{neal2011mcmc}.

\vspace{-.1in}
\subsection{Gibbs sampler details}
Corollaries~\ref{corr:psi_cond} and~\ref{corr:omega_cond} provide the basis of our Gibbs sampling algorithm. 
Each iteration of this algorithm starts with a pair $(\W_{\Psi \cup O}, \Psi)$, and repeats two steps: simulate a new Poisson grid $\Psi^*$ given $(\W_{\Psi \cup O}, \Psi)$, and then simulate a new set of diffusion values $\W^*_{\Psi^* \cup O}$ given $\Psi^*$. 
Recall that the set $O$ includes the start and end times, $0$ and $T$.
There are a few issues that must be resolved to translate these into a practical algorithm. 
We detail these below. \\
\textbf{Simulating a new Poisson grid $\Psi^*$ conditioned on $\W_{\Psi\cup O}$:} 
    Corollary~\ref{corr:psi_cond} shows that conditioned on the entire trajectory $\W$, $\Psi$ is a Poisson process with rate $\{M-\phi(\W_t), t \in [0,T]\}$. 
    In practice, our sampler will only evaluate $\W$ on the current set of Poisson times $\Psi$ and on the observation times $O$. 
    To simulate the new times $\Psi^{*}$, we exploit two facts: i) that the SDE skeleton summarizes the entire trajectory, whose values at other times can be retrospectively simulated from a Brownian bridge (steps~9 and 10 in algorithm~\ref{alg:EA1}), and ii) that $M-\phi(\cdot) \le M$. 
    We will use these along with the thinning theorem to simulate from the rate $M-\phi(\W)$ inhomogeneous Poisson process. 
    We first simulate a random set of times $\Gamma$ from a rate-$M$ Poisson process, and uncover $\W_\Gamma$, the trajectory on this set of times. 
    This second step just involves simulating from Brownian bridges over intervals defined by successive elements of $\Psi \cup O$ (algorithm~\ref{alg:EA1}, steps 9 and 10). 
    Having imputed $\W$ on $\Gamma$, we keep each element $g \in \Gamma$ with probability $1-\phi(\W_g)/M$, else we discard it. 
    The set of surviving elements of $\Gamma$ is a realization from a rate $M-\phi(\W)$ Poisson process, and forms the new times $\Psi^*$. 
    Along the way, we have evaluated $\W_{\Psi^*}$, the trajectory on this set of times. 
    Finally, we discard the path evaluations on the old skeleton, since, under the new skeleton, these can easily be resampled (again,  from a Brownian bridge). 
    {{The first five panels in figure~\ref{fig:alg_demo} shows these steps, where for simplicity we have ignored observations. }} \\
\begin{algorithm}[h]
   \caption{One iteration of the proposed auxiliary variable Gibbs sampler for EA1 diffusions}
   \label{alg:gibbs}
  \begin{tabular}{l l}
  \textbf{Input:  } & \text{A distribution $\pi(\cdot)$ over $X_0$, the initial value of the SDE} \\ 
    & \text{The drift term $\alpha(\cdot)$, and the associated quantities $A(\cdot), \phi(\cdot)$ and $M$} \\ 
                      & \text{The Poisson times $\Psi$ and the corresponding path values $\W_{\Psi}$} \\
                      & \text{The SDE path values $\W_O$ on the observation times $O$ (recall $O$ includes $0$ and $T$).} \\
                  \textbf{Output:  }& \text{A new SDE skeleton $(\Psi^{*}, \W_{\Psi^*}^{*})$, and new path values on $O$, $X^*_O$}.\\
   \hline
   \end{tabular}
   \begin{algorithmic}[1]
     \State {Simulate $\Gamma$ from a rate-M Poisson process on $[0, T]$}.
     \State Define $G = \Psi \cup O$. Write its $i$th element as $g_i$, with $g_1 = 0$ and $g_{|G|}=T$.
     \For {i in $1$ to $|G|-1$}
     \State Define $\Gamma_i = \Gamma \cap (G_{i}, G_{i+1})$. 
     Impute $\W$ on $\Gamma_i$ from a Brownian bridge: 
     \State \hspace{2in} $\W_{\Gamma_{i}}\sim \brobri_{\Gamma_{i}}(t_i,\W_{t_i},t_{i+1},\W_{t_{i+1}}).$
     \EndFor 
     \State Discard each point $g \in \Gamma$ with probability  $\frac{\phi(\W_g)}{M}$.
     Write $(\Psi^{*}, \W_{\Psi^{*}})$ for the set of surviving times and the associated path values.
    \State 
    Discard everything other than $\Psi^{*}, \W_{\Psi^{*}}$ and $\W_O$. 
     \State {Update $(\W_{\Psi^{*}}, \W_O) \equiv \W_{\Psi^* \cup O}$ on $\Psi^{*} \cup O$ with a Markov kernel having stationary distribution 
\begin{align*}
  p(\W_{\Psi^*},\W_O) &\propto \prior(\W_0)h_{\W_0}(\W_T) \text{BB}(X_{O \cup \Psi^*}|X_0,X_T) \ell(\W_O) \prod_{g \in \Psi^*} \left( 1- \frac{\phi(\W_g)}{M}\right) 
  \text{\ \ (see eq.\ \eqref{eq:hmc_target}).}
\end{align*}
We use Hamiltonian Monte Carlo. Write the new values as $(X^*_{\Psi^*}, X^*_O)$. Return $(\Psi^*,X^*_{\Psi^*},X^*_O)$.}
   \end{algorithmic}
\end{algorithm}
\noindent\textbf{Updating $\W$ conditioned on $\Psi$:} Corollary~\ref{corr:omega_cond} shows that conditioned on the Poisson grid $\Psi^*$, $\W_{\Psi^* \cup O}$ has density given by equation~\eqref{eq:hmc_target}.  
    This distribution, while intractable, can be evaluated up to a normalization constant, and is thus amenable to standard MCMC techniques 
   that update $\W_{\Psi^* \cup O}$ using a Markov kernel with equation~\eqref{eq:hmc_target} as stationary distribution. 
    We carry out this update using Hamiltonian Monte Carlo~\citep{neal2011mcmc}. 
    We can exploit the Markov structure of Brownian motion to calculate the log-likelihood and its gradient in linear time (see the experiments and appendix for details). 
    HMC exploits this gradient information to efficiently explore the conditional distribution.
    At the end of this step, we have a new set of path values $(\W^*_{\Psi^*}, \W^*_O)$. 
    This is shown in the last panel of figure~\ref{fig:alg_demo}. 
    Again, we can impute the SDE path $\W^*$ at any other set of times from a Brownian bridge.
Algorithm~\ref{alg:gibbs} outlines one iteration of our Gibbs sampler.

\begin{figure}[]
  \centering
  \includegraphics[width=.32\textwidth]{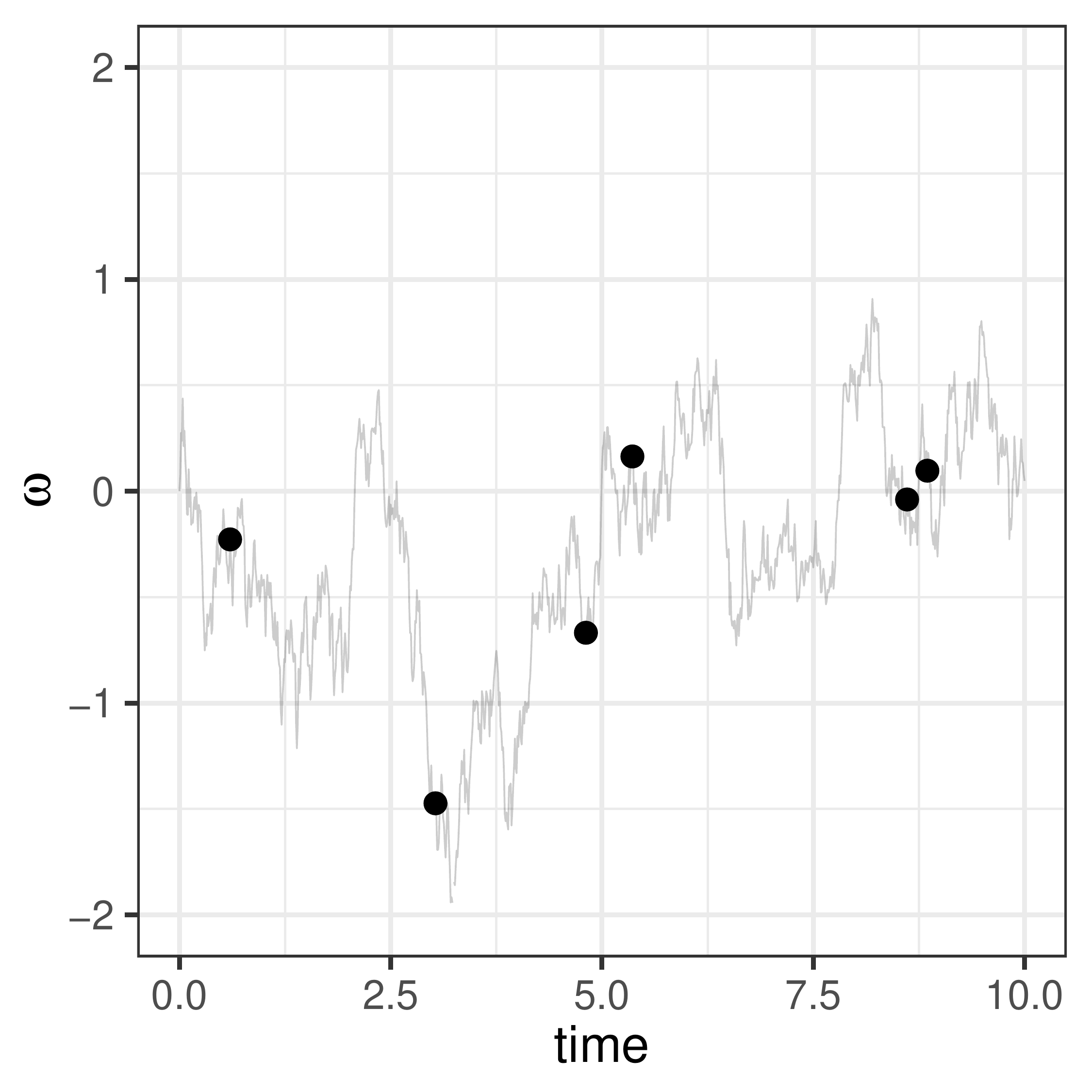}
  \includegraphics[width=.32\textwidth]{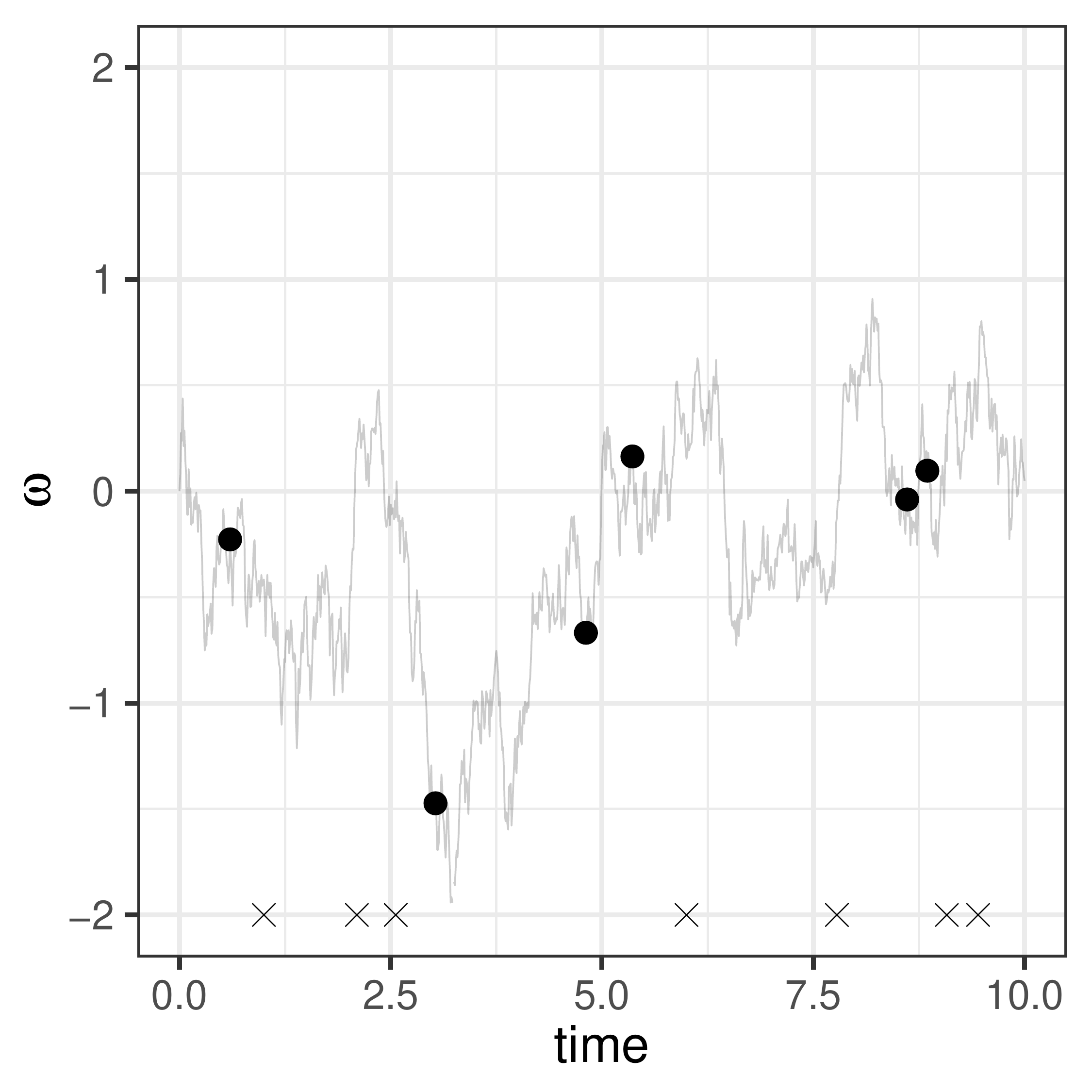}
  \includegraphics[width=.32\textwidth]{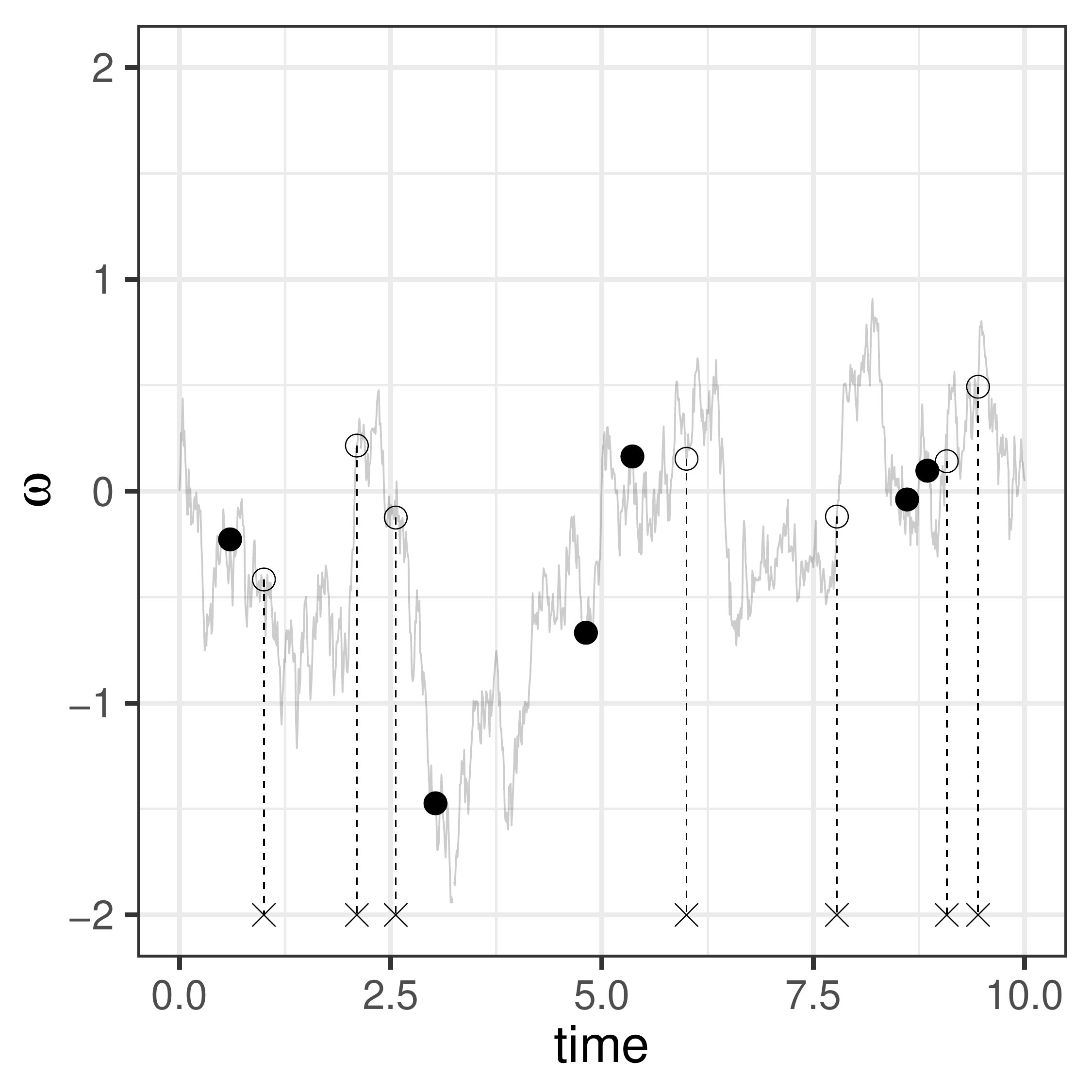}
  \includegraphics[width=.32\textwidth]{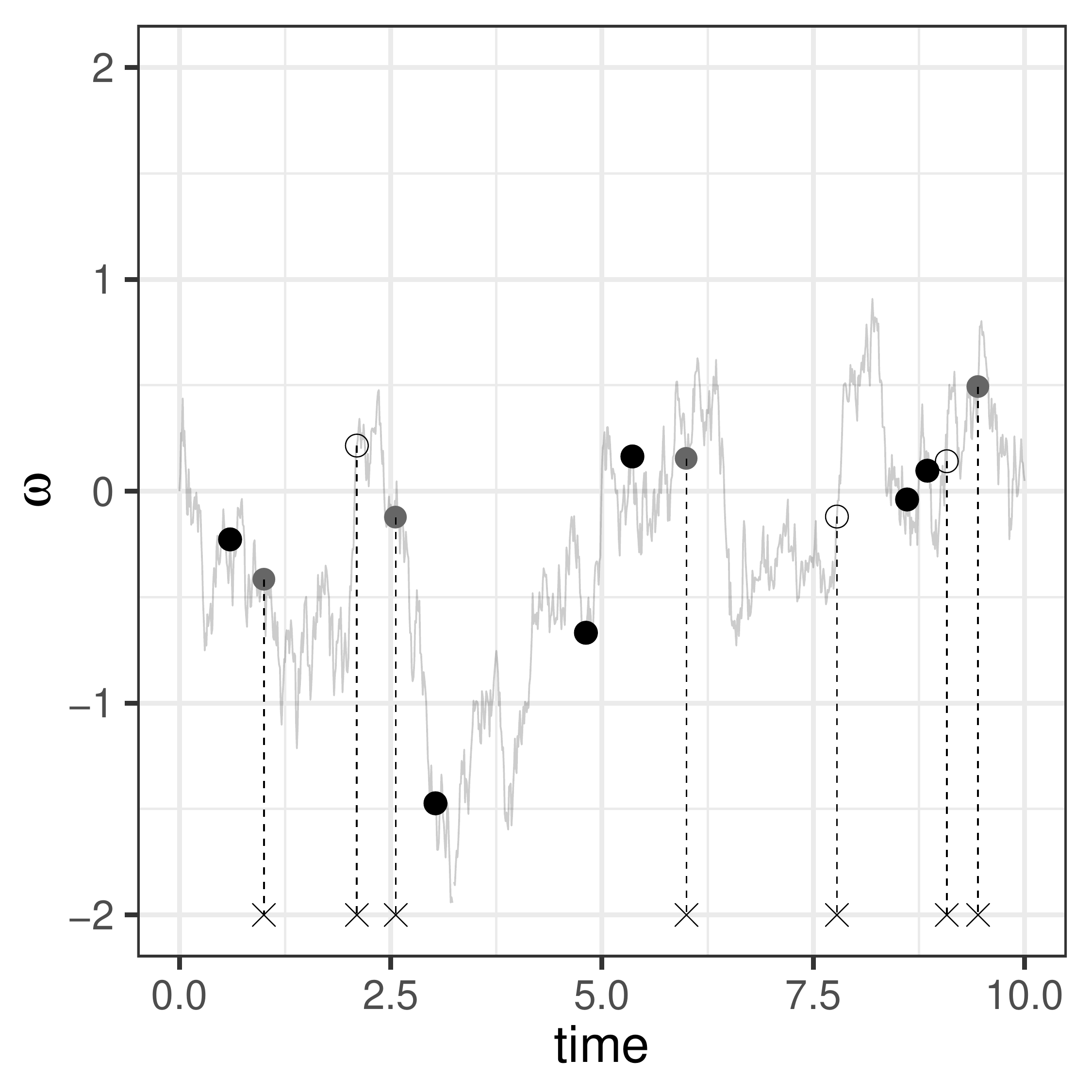}
  \includegraphics[width=.32\textwidth]{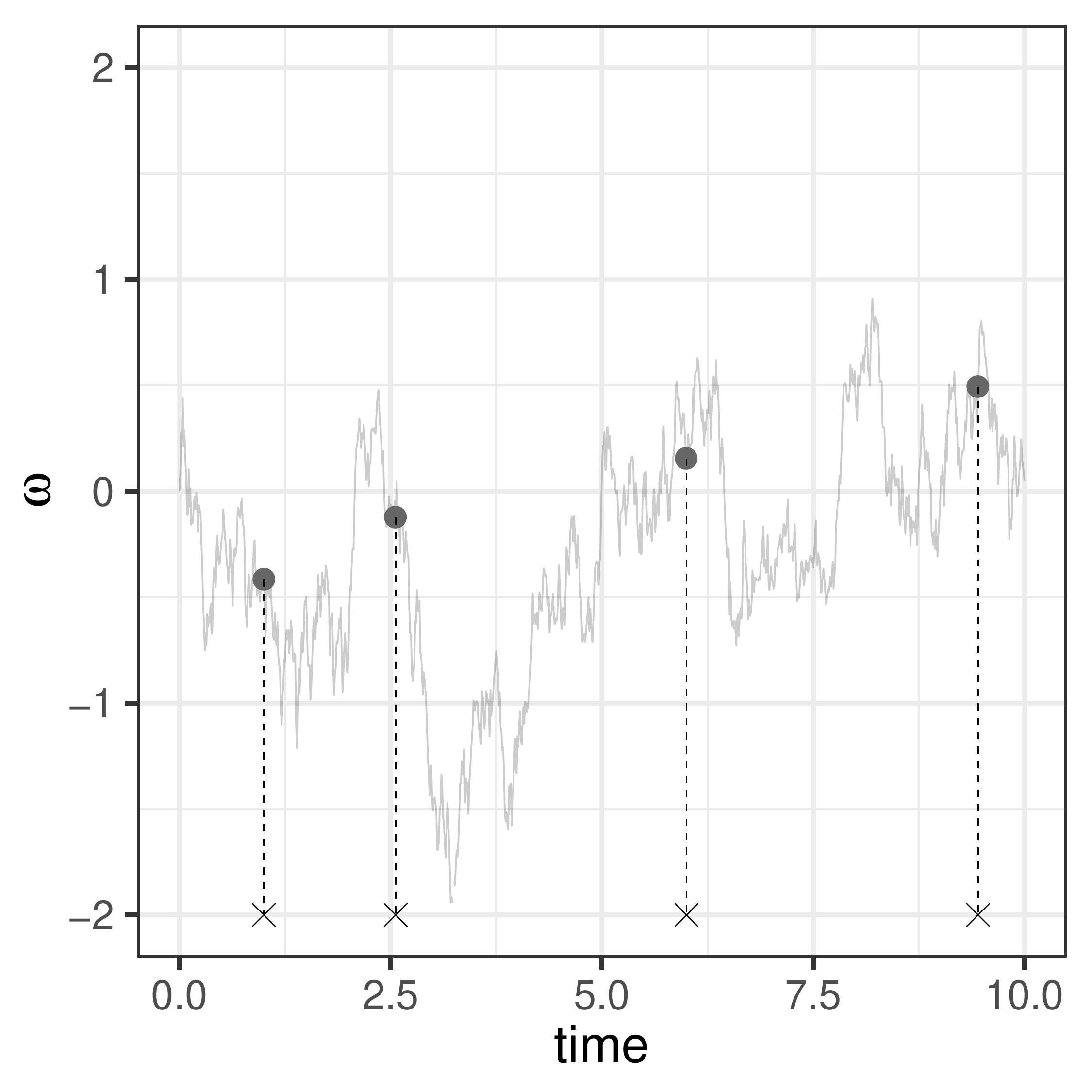}
  \includegraphics[width=.32\textwidth]{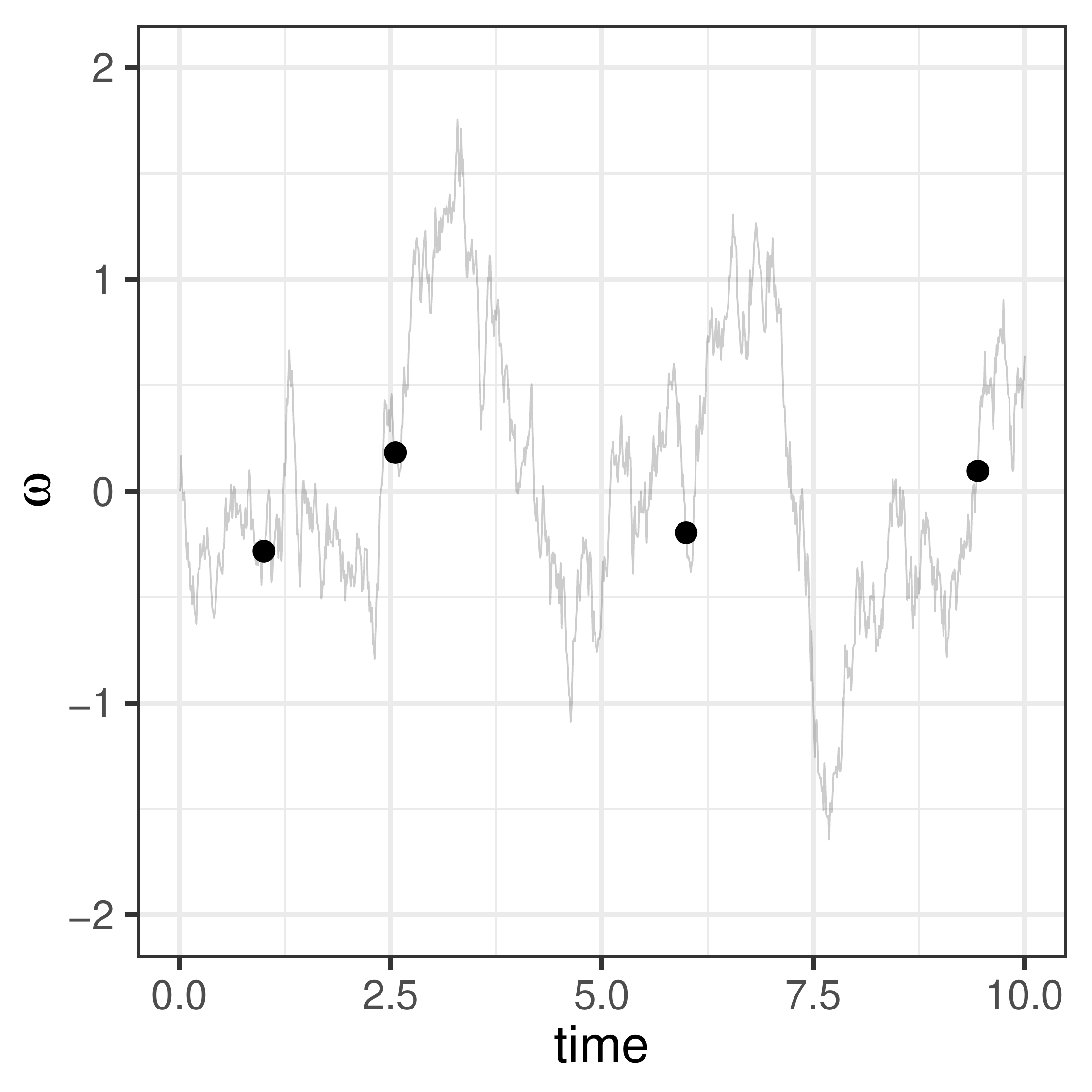}
  \caption{
    One interation of our proposed Gibbs sampling algorithm. For simplicity, we do not include observations (see Algorithm~\ref{alg:gibbs} for the general case).
    From the top-left to bottom-right: 
    1) The iteration starts with Poisson times $\Psi$ and the corresponding path values $\W_{\Psi}$. 
    This SDE skeleton is represented with the bold dots. 
    Also shown in grey is the SDE path. 
    This has not been instantiated by the algorithm, but can easily be simulated at any finite set of times from a Brownian bridge.
    2) Simulate $\Gamma$ from a rate-$M$ Poisson process on $[0, T]$ (shown as crosses). 
    3) Uncover $\W_\Gamma$, the SDE on $\Gamma$, by simulating from a  Brownian bridge (shown with hollow circles).
    4) Discard each element $g \in \Gamma$ with probability $\frac{\phi(\W_t)}{M}$. The points marked for deletion are kept hollow, while the surviving points filled in.
  5) Discard the hollow dots and the original skeleton $(\Psi,X_\Psi)$. The remaining times and values form the new skeleton, write this as $(\Psi^*,\W_{\Psi^*})$.
    6) Update $\W_{\Psi^*}$ via some standard MCMC kernel such as Hamiltonian Monte Carlo. The rest of the trajectory has also been refreshed, and can be simulated from a Brownian bridge.
  }
  \label{fig:alg_demo}
\end{figure}

For completeness, we include the following theorem which states that our sampler targets the joint measure $\ZZ^+$. Its proof is immediate (see~\citet{meyn2009}): the sampler has $\ZZ^+$ as its stationary distribution since the two Gibbs steps update the conditionals of $\ZZ^+$. The sampler is irreducible under mild conditions on the Markov kernel used to update $\W_{\Psi^* \cup O}$ given $\Psi^*$. 
\begin{theorem}
  The Gibbs sampler described above results in a Markov chain on the state space $(\Psi, \W)$ with stationary distribution $\ZZ^+(\Psi, \W) $.
\end{theorem}
\begin{proof}
  This follows immediately from the fact that two steps of the Gibbs sampler target the conditional distributions of $\ZZ^+(\Psi,\W)$.
\end{proof}

\section{Parameter inference} \label{sec:param}
Following~\citet[][Section 9]{beskos2006exact}, we extend our methodology to include posterior inference over the parameter $\theta$ in equation~\eqref{eq:sde_general}.
Our scheme absorbs the earlier trajectory update into a larger Gibbs sampler that also updates $\theta$ given the trajectory.
We start with equation~\ref{eq:joint_unnorm}, making explicit normalization constants that depend on $\theta$.
Recall $\QQ_x$ and $\WW_x$ are measures over paths starting at $x$ under the SDE and Brownian motion.
Let $\WW_{x,y}$ correspond to the Brownian bridge joining $\W_0=x$ to $\W_T=y$.
Integrating out the path between $0$ and $T$ gives the transition density 
\begin{align}
  p(\W_T=y|\W_0=x,&\theta) = N(y|x,T)\mathbb{E}_{\WW_{x,y}}\left[
  \exp\left\{ A_\theta(y) - A_\theta(x) -\frac{1}{2}
\int_0^T \left(\alpha_\theta^2(\W_t) + \alpha_\theta'(\W_t) \right) \df t \right\} \right] \nonumber \\ 
  &= N(y|x,T) \exp\left\{ A_\theta(y) - A_\theta(x) - 
    {L_\theta T}\right\} 
    \mathbb{E}_{\WW_{x,y}}\left[\exp\left\{ -\int_0^T \phi_\theta(\W_t)  \df t \right\}\right]
                 . 
\end{align}
On the other hand, from equation~\eqref{eq:joint_data}, conditioned on its endpoints, the trajectory and Poisson events have density with respect to the measure $\WW_{x,y}\times \mathbb{M}$ given by
\begin{align}
  p(\W, \Psi| \W_0=x, \W_{T}=y, \theta) &\propto  
   \exp(-M_\theta T) \prod_{g \in \Psi} \left( M_\theta - {\phi_\theta(\W_g)}\right) 
\end{align}
From Campbell's theorem, the normalization constant, obtained by integrating out both $\W$ and 
$\Psi$ is just the last term in the earlier equation.
Multiplying the above two equations, we get
\begin{align} 
  p(\W, \Psi| \W_0, \theta) &= N(\W_T|x,T) 
  \exp\left\{ A_\theta(\W_T) - A_\theta(\W_0) - 
  (M_\theta + {L_\theta })T\right\} 
  \prod_{g \in \Psi} \left( M_\theta- {\phi_\theta(\W_g)}\right).  \label{eq:param_joint}
\end{align}
Given observations at times $O$ in an interval $[0,T]$, we break the interval into segments $[o_{i-1},o_i]$, calculating the transition density across each segment as above.
The total density is the product of these. With a prior $p(\theta)$ over $\theta$, and dropping terms that do not depend on $\theta$, we have the following expression for the posterior over $\theta$ (see also Theorem 3 in~\citet{beskos2006exact}):
\begin{align}
  p(\theta|\psi, \W) \propto & p(\theta)  \ell_\theta(X_O)
  \exp\left\{ A_\theta(\W_T) - A_\theta(\W_0) - 
  (M_\theta + {L_\theta })T\right\} \nonumber \\
    & 
    \pi(X_0) \prod_{i=2}^{|O|} N(X_{o_i}|X_{o_{i-1}}, (o_i - o_{i-1}))
    \prod_{g \in \Psi} \left( M_\theta- {\phi_\theta(\W_g)}\right) 
\end{align}
Above, we allow the likelihood $\ell(\cdot)$ to also depend on $\theta$.
Simulating from this distribution is straightforward, and we do this using a Metropolis-Hastings step.
Our overall Gibbs sampler then alternates the two steps of algorithm~\ref{alg:gibbs} with a step to update $\theta$.
We point out that following ideas from~\citet{beskos2006exact}, we can use {\em non-centered reparametrizations}~\citep{papas2007general} that reduce coupling 
between diffusion paths and the parameter $\theta$. 
This is especially important when $\theta$ also affects the diffusion term $\sigma$: this situation requires some care, and we refer the reader to~\citet{beskos2006exact} for more 
details.

\vspace{-.3in}
\section{Related work}
\label{sec:related}
\vspace{-.1in}
Traditional approaches to simulating from an SDE involve time-discretization methods like the Euler-Maruyama method or Millstein's method. 
Time-discretization also simplifies posterior simulation, opening up the vast literature on MCMC sampling for discrete-time time-series models. 
Example methods include particle MCMC~\citep{andrieu2010particle}, the embedded HMM~\citep{neal2004inferring}, Hamiltonian Monte Carlo~\citep{neal2011mcmc} among many others. 
Discrete-time approximations however introduce bias into the simulations, and characterizing their effect in hierarchical models is not easy.
This makes it necessary to work with fine grids, resulting in long time-series and expensive computation. 
Further, controlling bias in this manner uncovers more of the diffusion, increasing coupling and degrading MCMC mixing~\citep{liu1994fraction, roberts2001inference}.

There are few approaches towards {\em exact} or {\em unbiased} estimation for diffusions to eliminate discretization error.
As described in subsection~\ref{sec:ea1}, our approach builds on a line of work starting from~\citet{beskos2005exact}, who proposed a rejection sampling algorithm allowing exact simulation from the EA1 class of SDEs.
%
%
Section~\ref{sec:bayes_model} shows how this prior similation method can be extended to posterior simulation given noisy observations.
Like our method, this involves instantiating the diffusion skeleton (the Poisson times and associated diffusion values), as well as the diffusion values on observation times. 
However, as we described, this algorithm alternately updates the diffusion skeleton given the values at observation times, and vice versa. 
By contrast, our algorithm updates the {\em entire} set of path values given the Poisson times, and then Poisson times given path values, reducing the coupling between the Gibbs steps.
Furthermore, this extension still involves the EA1 rejection sampling algorithm, and can have high rejection rates.
Controlling this requires instantiating more of the diffusion on additional grid points~\citep{beskos2006retrospective}, which will slow down mixing.
Our MCMC algorithm does not face this problem. 

In~\citet{fearnhead2008particle}, the authors propose another unbiased discretization-free algorithm, a random-weight particle filter to approximate the posterior distribution. 
This is a sequential Monte Carlo algorithm that targets the augmented distribution in equation~\eqref{eq:param_joint}. 
This algorithm is consistent as the number of particles tends to infinity, and for finite number of particles, can be incorporated into a particle MCMC scheme, giving an MCMC algorithm that targets the posterior without any error. 
There have been a number of follow-up papers improving~\citet{fearnhead2008particle}, whether by devising better proposal distributions or by developing particle smoothing algorithms that improve effective sample sizes and allow parameter inference~\citep{olsson2011particle, gloaguen2017online}. 
%
%
Another line of work for unbiased estimation with SDEs~\citep{rhee2015unbiased} builds on {\em multi-level Monte Carlo} (MLMC) methods~\citep{giles2008multilevel}.  
These methods involve picking a random time-discretization granularity, so that the interval $[0,T]$ is uniformly split into $2^g$ subintervals for a random $g$. 
With some care, the resulting algorithms allow unbiased estimation of path functionals of the SDE, and can be used for posterior estimation~\citep{Jasra2020}. 
These methods, along with some of the earlier particle methods, have the advantage of being applicable to a wider class of SDEs than we considered here, in particular they do not require the availability of a Lamperti transformation, and thus apply to more general multi-dimensional diffusions. 
%
Our HMC based-approach is quite different from these, and an interesting line of work is to use ideas from each to improve the other.

\vspace{-.2in}
\section{Experiments}
\vspace{-.1in}
\label{sec:exp}
In the following, we evaluate our sampler and a number of baselines on synthetic and real datasets. 
Our first baseline is the EA1 rejection sampling algorithm of~\citet{beskos2005exact}, we use this in settings where we want to simulate from an SDE prior, allowing us to study trade-offs between producing cheap but dependent samples from our MCMC algorithm, and producing independent samples at the possible cost of high rejection rates. 
Our second baseline is an approximate Markov chain Monte Carlo sampling algorithm, and uses Euler–Maruyama discretization to construct a particle Markov chain Monte Carlo (pMCMC) sampler. 
pMCMC~\citep{andrieu2010particle} is a standard and relatively off-the-shelf tool to simulate from nonlinear hidden state-space models. It makes proposals from a particle filtering algorithm, which are then accepted or rejected with appropriate probability. 
Algorithm~\ref{alg:partMCMC} in the appendix outlines the details of the algorithm,
we considered time-discretization levels of $0.1$ and $0.01$. 
Our last baseline is the unbiased random-weight particle filter of~\citet{fearnhead2010random}. 
We considered both this particle filter, as well as an exact pMCMC algorithm based on it. 
For both pMCMC algorithms, we considered a variety of settings for the number of particles, reporting results with 50 particles:
this usually gave best performance, with run-time  becoming unmanageably long with more than 200 particles.
We will refer to the time-discretizated pMCMC algorithm as \eul, the random-weight particle filter as \pf, and the exact pMCMC algorithm as \fearn.
All experiments were carried out on a desktop with an Intel(R) Core(TM) i7-3770 CPU @ 3.40GHz and 16GB RAM.

\vspace{-.1in}
\subsection{Example 1: The hyperbolic bridge}
\label{sec:hyp}
Consider the hyperbolic bridge, a special case of the hyperbolic diffusion~\citep{barndorff1978hyperbolic}: 
\begin{align} \label{eq:hyper1}
  \df X_t = -\frac{\theta X_t}{\sqrt{1+X_t^2}} \df t + \sigma \df B_t,  \quad \theta > 0.
\end{align}
As stated in Section~\ref{sec:sdes}, we fix the diffusion parameter $\sigma$ to 1. When we are not updating $\theta$, we fix it to $1$. 
It is easy to verify that the drift $\alpha(x) = -\frac{\theta x}{\sqrt{1+x^2}}$ satisfies the assumptions of Girsanov's theorem. We can calculate $A(x) = \int_0^x \alpha(u) \df u = \theta-\theta \sqrt{1 + x^2}$ and $\alpha'(x) = -\frac{\theta}{(1+x^2)^{3/2}}$, showing that $\frac{1}{2}(\alpha^2(x) + \alpha'(x)) =  \frac{1}{2}(\frac{\theta^2 x^2}{1+x^2} - \frac{\theta}{(1+x^2)^{3/2}})$ lies in $[-\frac{\theta}{2},\frac{\theta^2}{2}]$. We set 
\begin{align} 
  \phi(x) := \frac{1}{2}(\alpha^2(x) + \alpha'(x)) + \frac{\theta}{2}= \frac{1}{2}\left(\frac{\theta^2 x^2}{1+x^2} - \frac{\theta}{(1+x^2)^{3/2}}\right) + \frac{\theta}{2}. 
\end{align} 
This lies in the interval $[0,\frac{\theta^2}{2} + \frac{\theta}{2}]$. Accordingly, the EA1 Poisson process intensity $M$ equals $\frac{\theta^2}{2} + \frac{\theta}{2}$, and the associated $h$-biased Brownian bridge has $h_x(\omega_T) \propto \exp{\left(-\theta\sqrt{1 + \omega_T^2} + \theta\sqrt{1 + x^2} - \frac{(\omega_T - x)^2}{2T}\right)}$.


\noindent\textbf{Tuning the HMC sampler:} A key step of our Gibbs sampler involves conditionally updating the SDE trajectory $X$ given the Poisson grid $\Psi$, following equation~\eqref{eq:hmc_target}. 
We implement a Markov kernel that targets this conditional distribution using Hamiltonian Monte Carlo (HMC)~\citep{neal2011mcmc}, a widely used MCMC algorithm. 
We provide more details of this in the appendix, at a high-level this requires computing the gradient of the log of the joint probability specified in equation~\eqref{eq:hmc_target}.
HMC requires tuning three parameters $M, N$ and $\epsilon$, corresponding respectively to a mass matrix, the number of leapfrog steps and the leapfrog stepsize. 
The latter two govern the leapfrog symplectic approximation to the Hamiltonian dynamics that HMC uses to update $\W$. 
We obtained best performance for $M$ between 10 to 100 times the identity matrix and chose the latter (see~\citet{neal1996sampling, beskos2011hybrid} for more sophisticated approaches to setting $M$). We tried a range of values for both the size $\epsilon$ and number $N$ of leapfrog steps ($\{0.1, 0.2, 0.5, 1, 2\}$ and $\{1, 2, 5, 10\}$ respectively). 
We evaluate these for three problems, corresponding to simulating the hyperbolic SDE on intervals with length $T$ equal to $10, 20$ and $50$. 
For each combination of $\epsilon, N$ and $T$, we produced $10000$ samples from our sampler. 
\begin{figure}[]
  \centering
  \includegraphics[width=.32\textwidth]{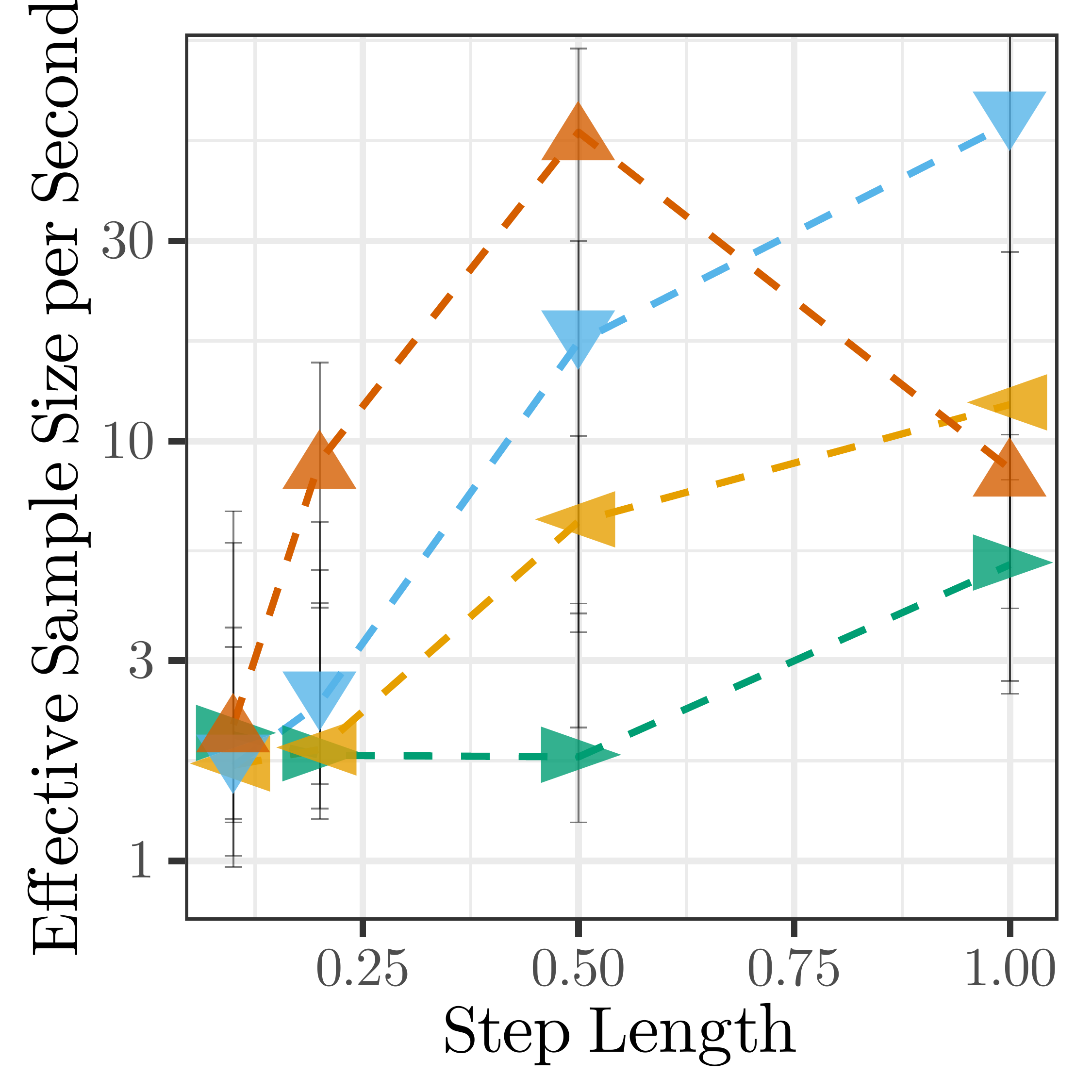}
  \includegraphics[width=.32\textwidth]{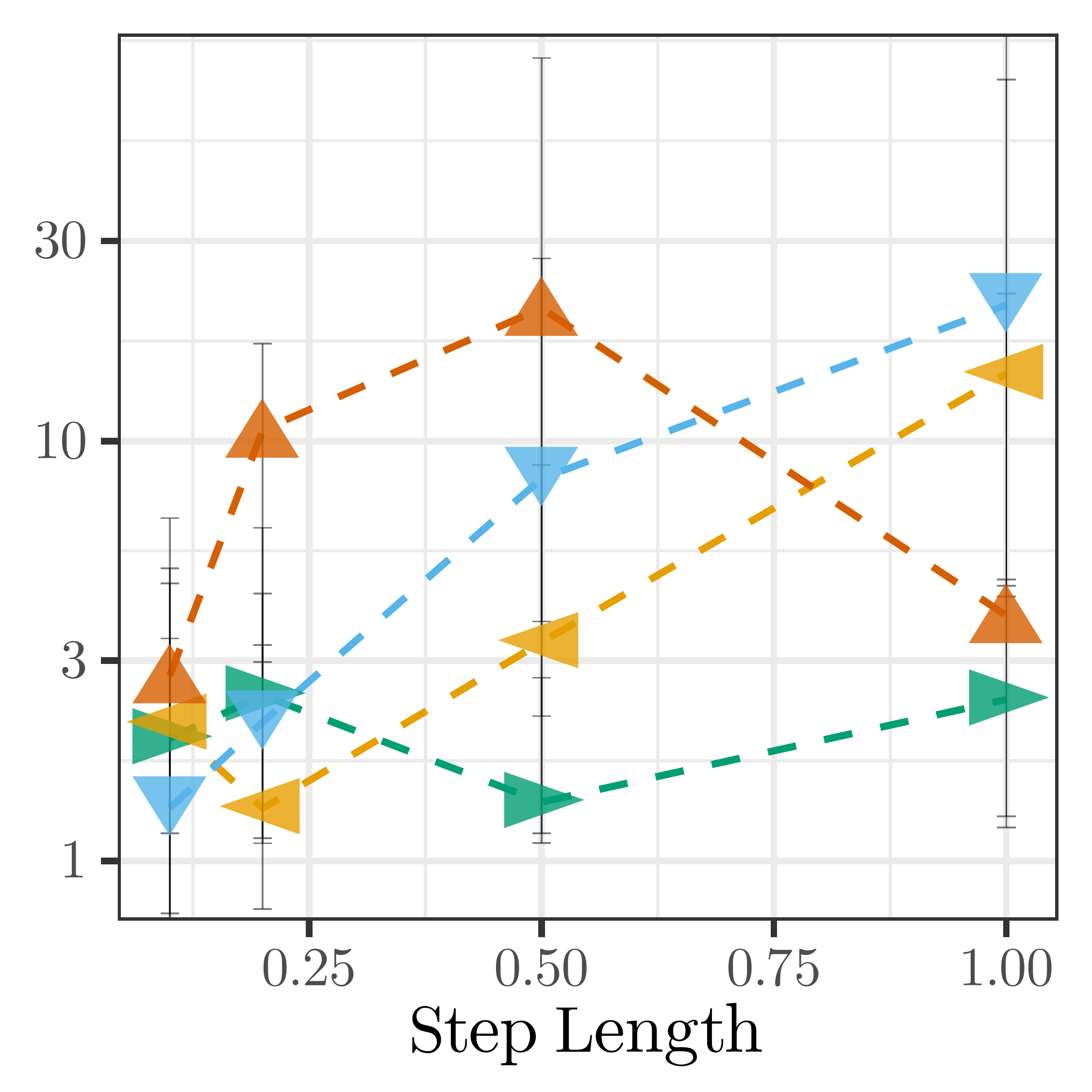}
  \includegraphics[width=.32\textwidth]{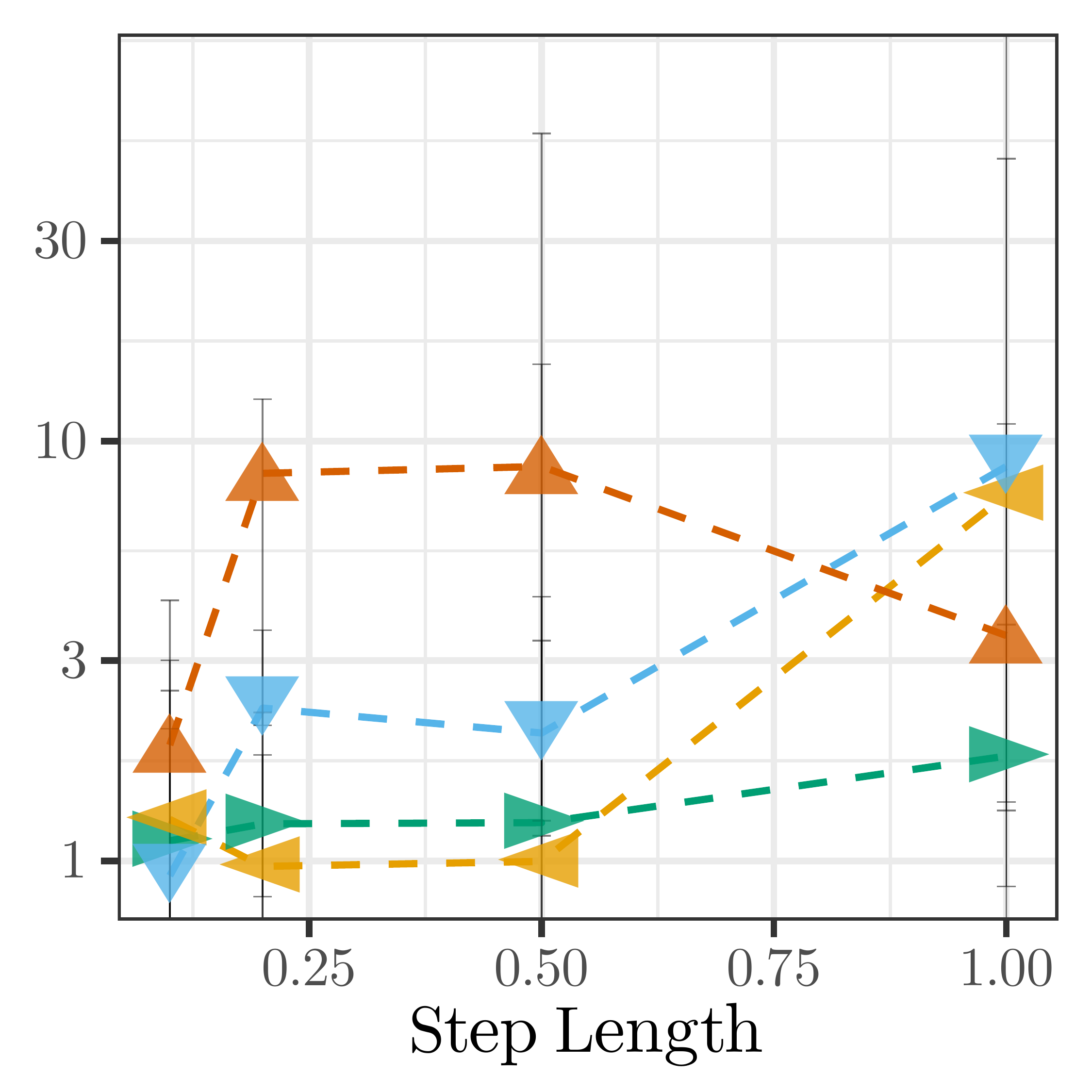}
  \caption{Effective sample size per second (ESS/s) for different settings of stepsize $\epsilon$ (x-axis), each curve being a different number of leapfrog steps $N$ for the HMC sampler. From left to right, the three panels fix $T = 10, 20$ and $50$ respectively. Different symbols represent the different values of $N$: {$\blacktriangleright$} represents 1 step, $\blacktriangleleft$ represents 2 steps,
  {$\blacktriangledown$} represents 5 steps and {$\blacktriangle$} represents 10 steps.}
  \label{fig:hyper_para}
\end{figure}

To evaluate sampler performance, we calculate the effective sample size (ESS) of $X_{T/2}$, the diffusion evaluated at $T/2$, the midpoint of the simulation interval. ESS estimates the number of independent samples that the MCMC output is equivalent to, and we calculated this using the \texttt{R} package \texttt{rcoda}~\citep{plummer2006coda}. 
To account for the different settings having different computational cost, we divide ESS by the compute time, yielding effective sample size per second (ESS/s) as our metric of sampler efficiency. 

From left to right, the three panels of figure~\ref{fig:hyper_para} fix $T$ to $10, 20$ and $50$, and plot ESS/s for different settings of $N$ and $\epsilon$. 
Based on this, we choose a fairly standard setting, with the stepsize $\epsilon$ equal to $0.2$ and number of steps $N$ equal to $5$. 
This configuration also performs adequately for other SDEs that we consider, moreover the algorithm does not show strong sensitivity to the choice of these tuning parameters. 

\noindent\textbf{Prior simulation:} Using these HMC parameters, we compare the efficiency of our sampler with two baselines, a simple approximation based on Euler-Maruyama discretization, and the exact EA1 rejection sampler. 
We used each method to simulate $10000$ trajectories from the hyperbolic diffusion. 
For each setting, we carried out 10 repetitions to produce error bars (not visible sometimes due to low variance), and plot ESS/s of $\W_{T/2}$ against $T$ in Figure \ref{fig:hyper_hmc_ea_euler}.

\begin{figure}[]
  \centering
  \begin{minipage}[hp]{0.4\linewidth}
  \centering
    \vspace{-0.1 in}
    \includegraphics[width=.99\textwidth]{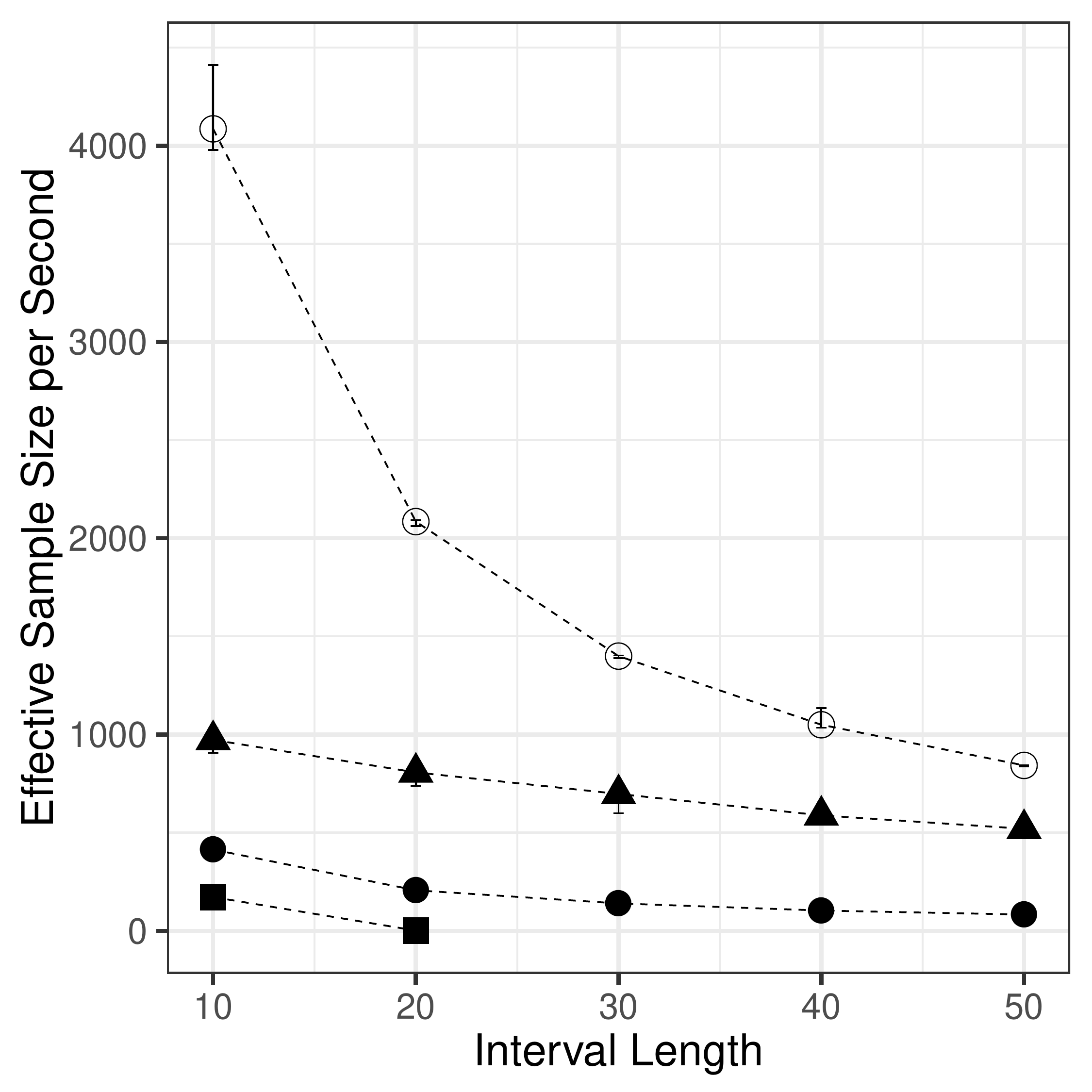}
    \vspace{-0.1 in}
  \end{minipage}
  \begin{minipage}[hp]{0.54\linewidth}
    \vspace{-0.6 in}
    \caption{ESS/s against simulation interval $T$ for the hyperbolic bridge prior. {\footnotesize $\blacktriangle$} represents our Gibbs sampler, $\bullet$ and $\circ$ represents Euler–Maruyama method when stepsize equals 0.01 and 0.1 respectively, and {\tiny $\blacksquare$} represents EA1. 
    Due to low acceptance rates, we did not run EA1 for interval lengths longer than 20. 
  }
  \label{fig:hyper_hmc_ea_euler}
    \vspace{-0.1 in}
  \end{minipage}
    \vspace{-0.4 in}
\end{figure}
Unsurprisingly, we observe that the Euler-Maruyama approximation with the coarsest time-discretization of $0.1$ is the most efficient algorithm computationally. 
Note though that this is an approximate algorithm. We can improve its accuracy by using a finer grid, a typical setting being a grid with resolution $0.01$. 
Interestingly, for this setting, even after correcting for the dependent samples, our MCMC algorithm is more efficient that Euler-Maruyama, with the Poisson grid allowing much fewer evaluations of the SDE trajectory. 
Additionally, the gap between our sampler and the $0.1$-grid Euler-Maruyama sampler reduces with $T$ through a combination of faster run-times and reduced dependency between MCMC samples. 
All algorithms were significantly more efficient that the exact EA1 algorithm, and for interval lengths greater than $20$, the acceptance rates (which decay exponentially with interval length) became too small to produce samples in a reasonable amount of time. 
As mentioned in Section~\ref{sec:bayes_model}, it is possible to reduce rejection rates by breaking the interval into smaller segments. 
However, noting the poor performance of the algorithm even for smaller intervals, we did not investigate this further. 

\noindent\textbf{Posterior simulation:} Our main interest is in settings where we wish to simulate from the diffusion conditioned on noisy measurements. 
We consider the following setting: additive Gaussian noise with mean $0$ and standard deviation $0.2$, at regularly spaced times on $[0,T]$. 
We compare with the two particle Markov Chain Monte Carlo algorithms, \eul\ and \fearn, running these with 50 particles and \eul\ with a discretization level of 0.01.
We considered two settings, first where we varied the length of the time interval $T$, keeping the number of observations fixed at $20$ (figure~\ref{fig:hmc_hyp_nObs_tEnd}, left panel), and second, where we varied the number of observations keeping $T=20$ (figure~\ref{fig:hmc_hyp_nObs_tEnd}, middle panel). 
All samplers were run for $10000$ iterations, and each setting was repeated 10 times to produce error bars. 
\begin{figure}[]
  \centering
  \begin{minipage}[hp]{0.32\linewidth}
  \centering
    \vspace{-0.1 in}
    \includegraphics[width=1\textwidth]{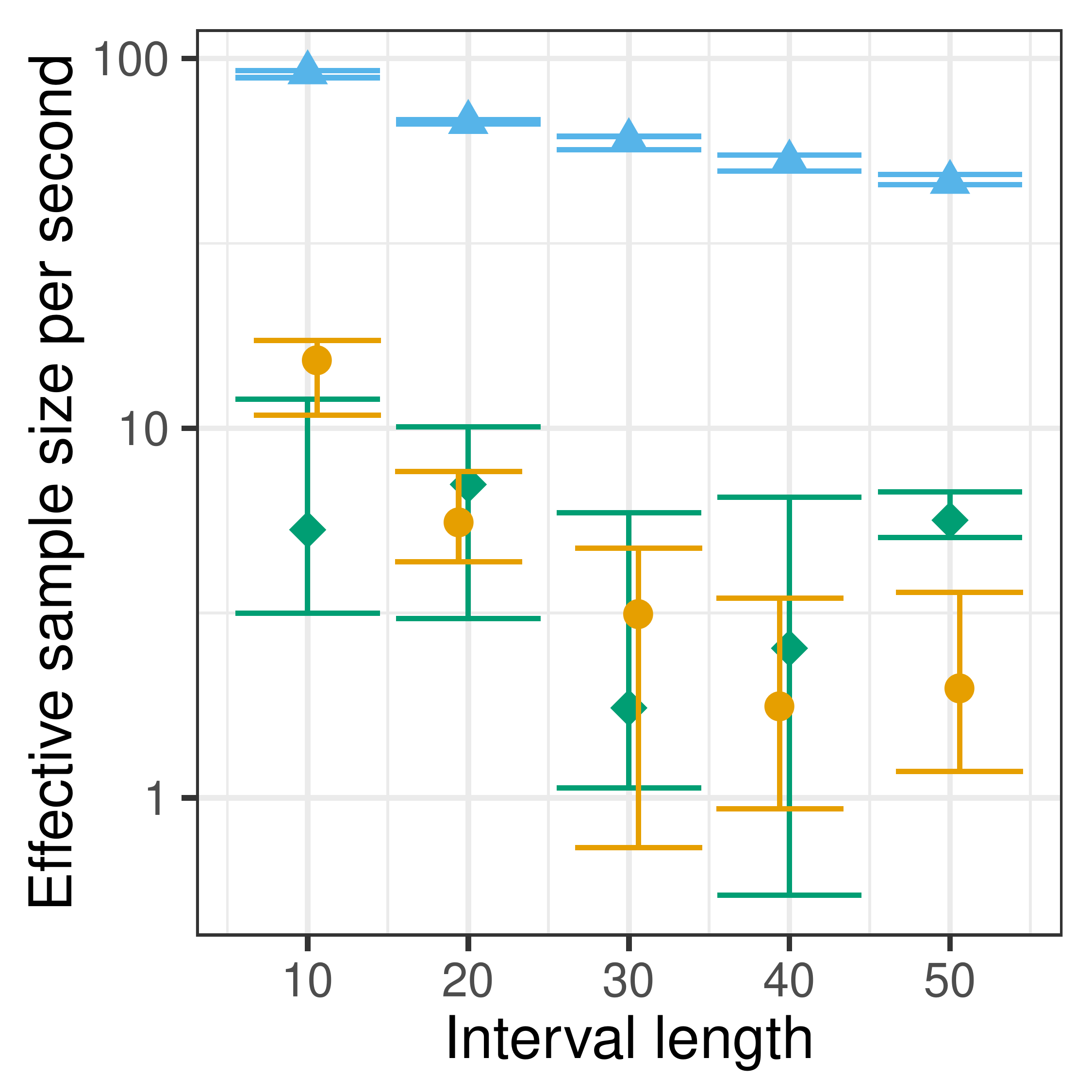}
    \vspace{-0.1 in}
  \end{minipage}
  \begin{minipage}[hp]{0.32\linewidth}
  \centering
    \vspace{-0.1 in}
    \includegraphics[width=1\textwidth]{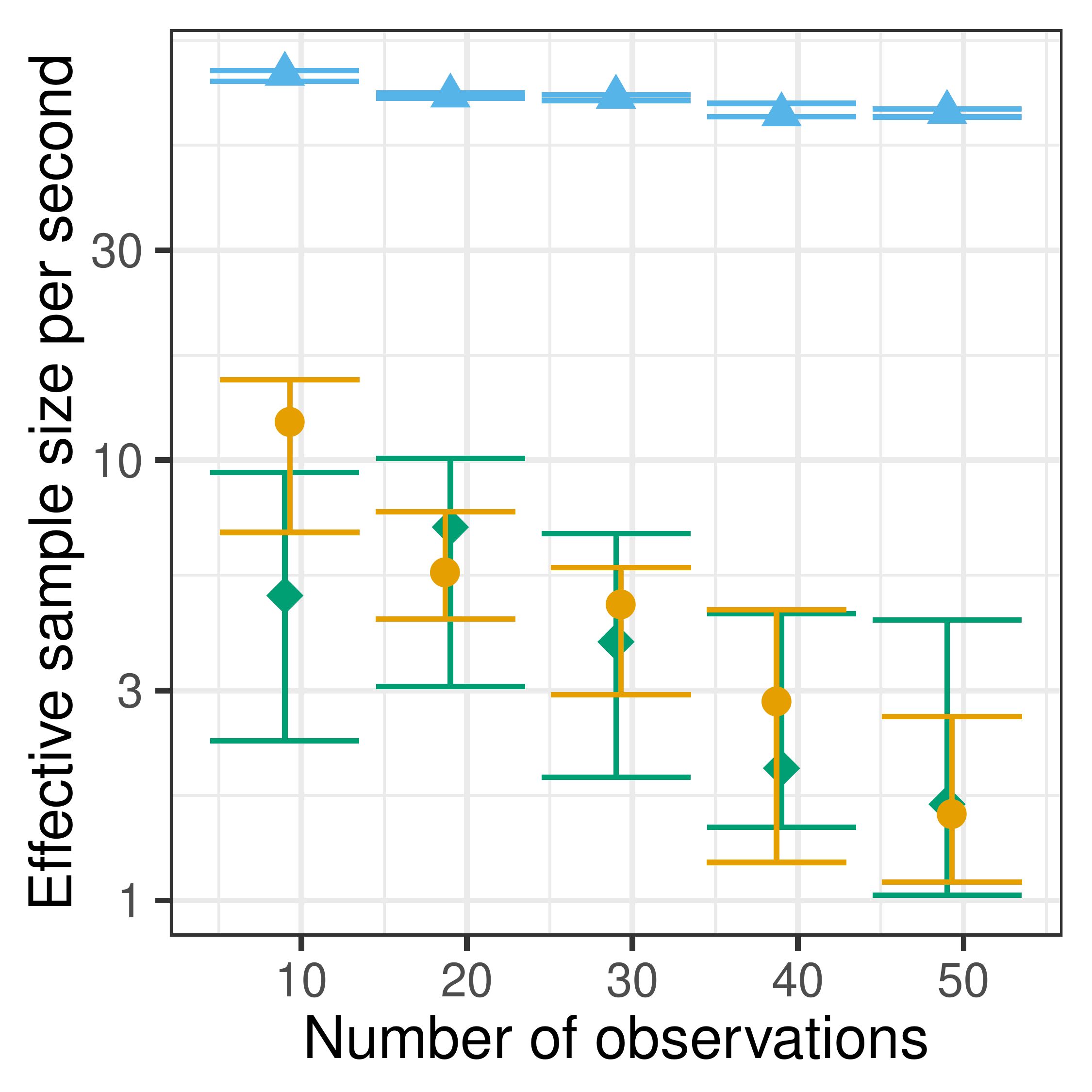}
    \vspace{-0.1 in}
  \end{minipage}
  \begin{minipage}[hp]{0.32\linewidth}
  \centering
    \vspace{-0.1 in}
  \includegraphics[width=1\textwidth]{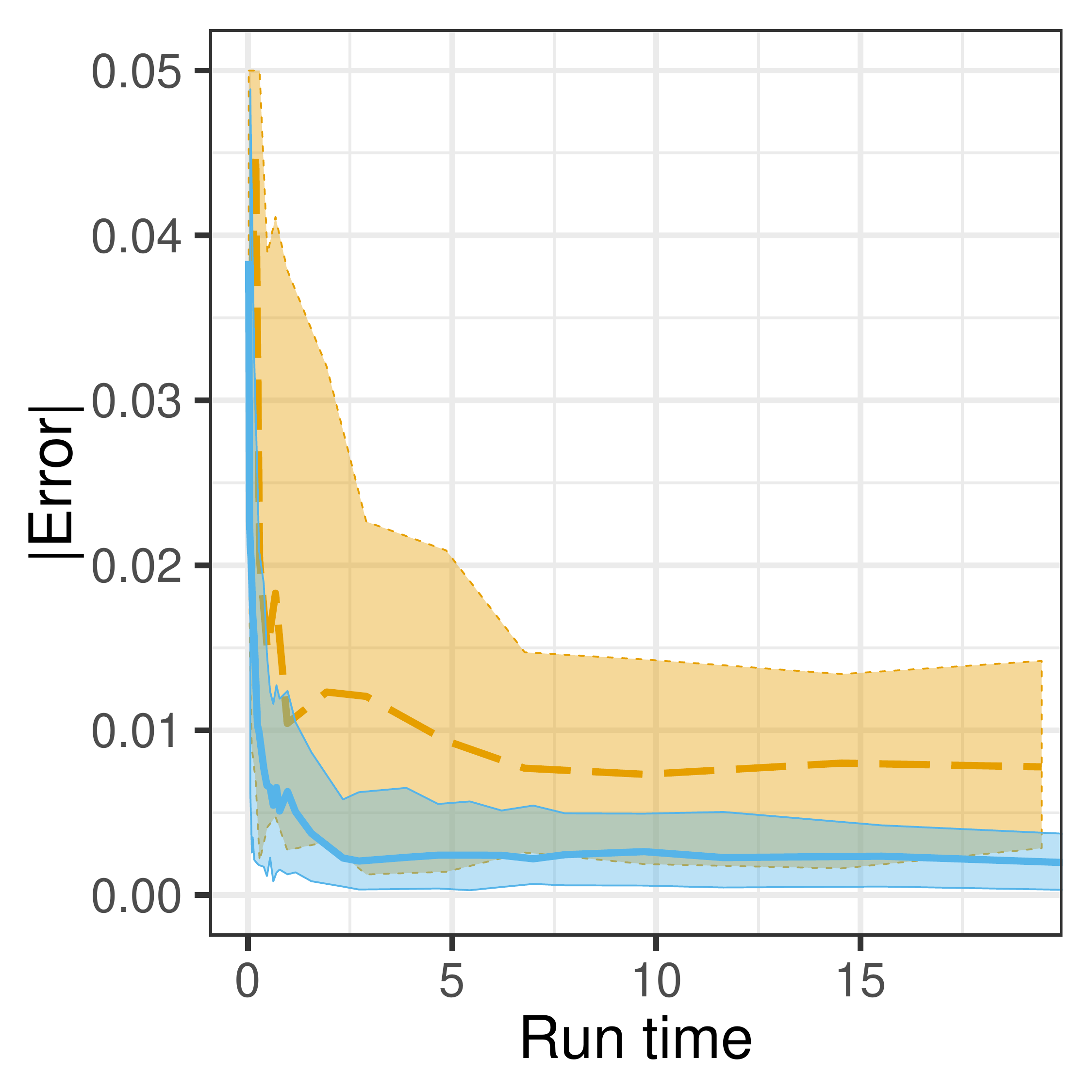}
    \vspace{-0.1 in}
  \end{minipage}
  \caption{ESS/s of our Gibbs sampler {\footnotesize $\blacktriangle$} and $50$-particle pMCMC samplers \fearn\ $\bullet$ and \eul\ {\footnotesize $\blacklozenge$} for (left) increasing interval $T$ with $N=20$ observations, and (middle) increasing $N$ with $T=20$. The right panel shows error in filtering distributions of our method (solid line) and particle filtering \pf\ (dashed line) for increasing computational budget.}
    \label{fig:hmc_hyp_nObs_tEnd}
\end{figure}
For both setups, we see that our method is more than an order of magnitude more efficient than both pMCMC algorithms. 
This performance gain increases as $T$ increases, where the increasingly long time-series both increase the run-time of the pMCMC algorithms, as well as reduce acceptance probabilities. 
As we note in section~\ref{sec:related}, it is possible to improve pMCMC performance with more careful choice of proposal distribution, though a more challenging issue is the long run-times involved with both pMCMC algorithms.


To address the last point, we note that \fearn\ is built on an unbiased particle filtering algorithm \pf. 
The latter forms a much cheaper baseline to compare against our MCMC sampler, with a run of 1000 particles taking slightly less time than 1000 iterations of our MCMC sampler. 
Unlike our MCMC sampler, the particle filter returns weighted samples, making comparison of the two algorithms trickier.
To understand the speed accuracy trade-offs involved, we used both algorithms in a filtering task:
 we simulated a diffusion over an interval of length 20, generating 20 equally spaced Gaussian observations centered on the path, and with standard deviation $0.2$.
We looked at the ability of both algorithms to reconstruct statistics of the posterior distribution at the last observation: $p(\W_N|y_1,\cdots,y_N)$. 
We considered the posterior mean, posterior variance as well as the posterior expectation of $\exp(\W_N)$.

The rightmost panel in figure~\ref{fig:hmc_hyp_nObs_tEnd} show the results for the absolute difference between the estimated variances and the `true' posterior variance, obtained from a long run of \eul. 
The plot shows the errors of our algorithm and \pf\ as computational budget (i.e.\ run-time) increases. The other two statistics gave similar results.
For our MCMC algorithm, increasing the computational budget just involved generating more samples (upto 5000 samples in the figure), for particle filtering, this involved running the algorithm with increasing number of particles (upto 5000 particles).
We can see that for the same budget, our method produces more accurate results, converging quickly to reasonable estimates that improve with more samples.
Although the random-weight particle filtering algorithm is consistent, it involves a significant amount of variance. 
Producing comparable errors to MCMC after many iterations will requires a large number of particles, raising issues with memory. 
We emphasize that particle filtering can also be improved in other ways, for example by choosing better proposal distributions, though we did not investigate this.
We note though that our filtering task was designed to favor particle filtering, and a more general smoothing task would require further extensions of the basic particle filter.


\begin{figure}[H]
  \centering
  \begin{minipage}[hp]{0.32\linewidth}
    \includegraphics[width=\textwidth]{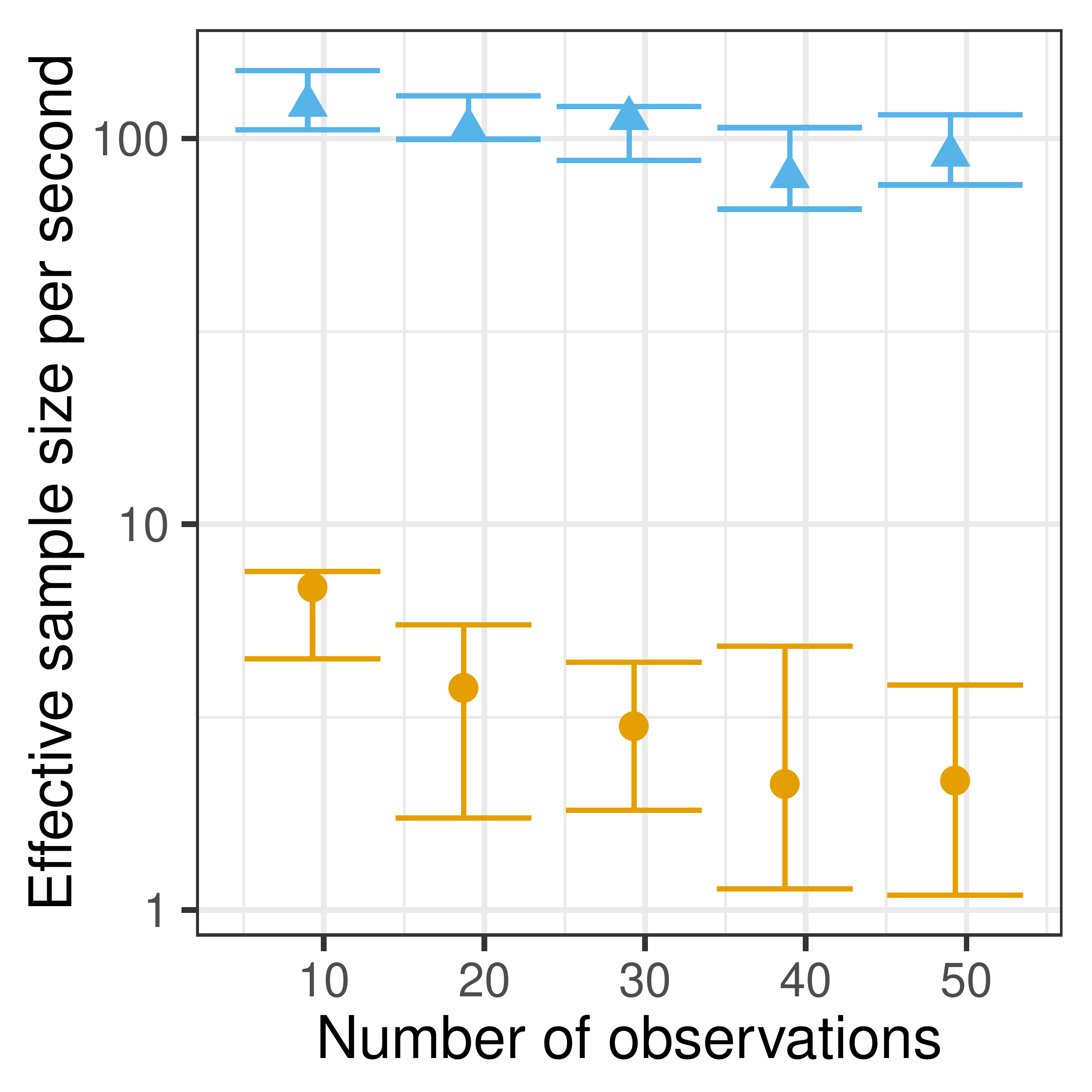}
  \end{minipage}
  \begin{minipage}[hp]{0.32\linewidth}
  \includegraphics[width=\textwidth]{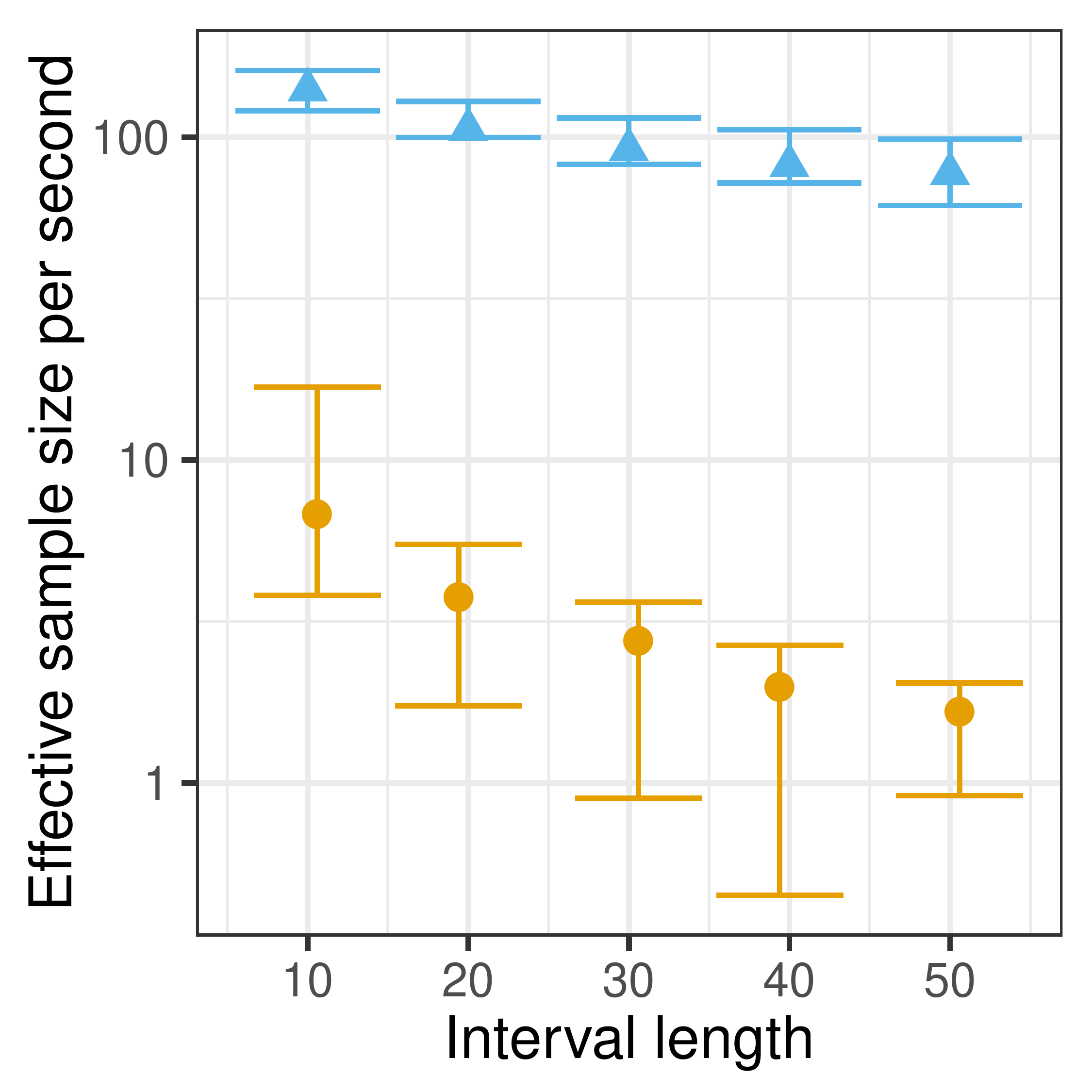}
  \end{minipage}
  \begin{minipage}[hp]{0.32\linewidth}
    \vspace{-.2in}
  \caption{ESS/s for posterior samples of $\theta$ for our sampler and a PIMH \eul\ sampler, as we increase number of observations $N$ with interval-length $T$ fixed at 20 (left), and as $T$ increases with $N=20$.}
  \label{fig:hmc_hyp_pars}
  \end{minipage}
\vspace{-.2in}
\end{figure}

In our final experiment, we place an exponential prior on the parameter $\theta$, and look at sampling from its posterior distribution.
We compare our MCMC sampler from section~\ref{sec:hyp} with \eul\ extended to include parameter inference.
This is a particle-independent Metropolis-Hastings (PIMH) sampler that proposes a new parameter $\theta^*$, and uses \eul\ to calculate the MH acceptance probability (see appendix).
Figure~\ref{fig:hmc_hyp_pars} shows the results using the same MH proposal distribution for both algorithms: we propose from the prior over $\theta$. 
Again our sampler is significantly more efficient than $\eul$.
We expect performance to further improve if we exploit ideas from~\citet{beskos2006exact} to reduce coupling between path and parameter, or jointly update path and parameter in the HMC step.
The PIMH sampler on the other hand does not suffer from such coupling, and its relative performance will improve when this is a significant factor.

%

\vspace{-.1in}
\subsection{Example 2: Periodic Drift}
\vspace{-.1in}
Our second example considers the sine-diffusion, an SDE with periodic drift $\alpha(x) = \sin(x-\theta)$:
\vspace{-.1in}
\begin{equation} \label{eq:sin1}
  \hspace{0.8in} \df X_t = \sin(X_t-\theta) \df t + \df B_t.
\end{equation}
Now $A(x) = \int_0^x \alpha(u) \df u = \cos(\theta) - \cos(x-\theta)$, and $h_x(u) \propto \exp(A(u) -A(x) - (u - x)^2/2T) = \exp(-\cos(u-\theta) + \cos(x-\theta)- (u-x)^2/2T).$ 
Now, $\sin^2(x) + \cos(x)$ lies in $[-1,5/4]$ and we set $\phi(x) = \sin^2(x-\theta)/2 + \cos(x-\theta)/2 + 1/2$. This lies in $[0,9/8]$, so that this SDE is of class EA1. 

The periodic drift term $\alpha(\cdot)$ in this SDE presents a potential challenge to our MCMC methodology due to the bimodality around zero when $\theta=0$. For positive values of $X_t$ in the interval $(0, \pi)$, the drift term is also positive, and the SDE experiences a repulsive push away from $0$. 
A similar effect, but in the opposite direction, occurs when $\W_t$ lies in $(-\pi,0)$. 
The symmetry of the problem means that $\W_t$ and $-\W_t$ are equally likely, however the repulsion away from $0$ can make it difficult for an MCMC algorithm to cross from one to the other. 
We can overcome this with a simple additional MCMC step: at the end of each iteration, flip the sign of the entire trajectory with probability $0.5$. 
This approach that exploits the problem's symmetry works well if we want prior samples from the sine-diffusion, producing similar results to figure~\ref{fig:hyper_hmc_ea_euler} (see the appendix). 

For posterior simulation given fairly informative observations, our experiments show that this bimodality in the prior presents less of a problem.
In settings where it might be a problem, the simple approach of flipping the path signs will require an MH correction step. 
The efficacy of this will depend on the degree of asymmetry introduced by the likelihood. 
In the appendix, we introduce a more general and flexible {\em tempering} scheme~\citep{swendsen1986replica,neal1996sampling} to explore the trajectory space more effectively, but do not discuss it in the main document.

Figure \ref{fig:hmc1_pmcmc1_nObs_tEnd} fixes $\theta=0$, and plots effective sample sizes per second of posterior samples of $\W_{T/2}$.
Again, the samplers are given equally spaced noisy observations of the diffusion trajectory, having mean equal to the trajectory value, and standard deviation equal to 0.2. 
In the left panel, we keep the number of observations fixed at $20$ as we vary the interval length $T$, while in the middle panel, we vary the number of observations with $T=20$. 
Once again, our sampler significantly outperforms the two 50-particle pMCMC baselines, \fearn\ and \eul\ (with a discretization level of $0.01$).
In all our experiments, we verified our samplers were exploring the posterior distribution by running a Kolmogorov-Smirnov two-sample test on outputs from different MCMC algorithms, and in all cases, the test failed to reject the null that the samples come from the same distribution.
This indicates that our sampler is not stuck in a mode of the posterior because of the bimodal drift function. 
The rightmost panel in the figure compares our sampler with the particle filter \pf\ on a filtering task, and we get results similar to the earlier experiment: for the same computational cost, our MCMC sampler is more accurate with less variance.
\begin{figure}[]
  \centering
  \begin{minipage}[hp]{0.32\linewidth}
  \centering
    \vspace{-0.1 in}
    \includegraphics[width=1\textwidth]{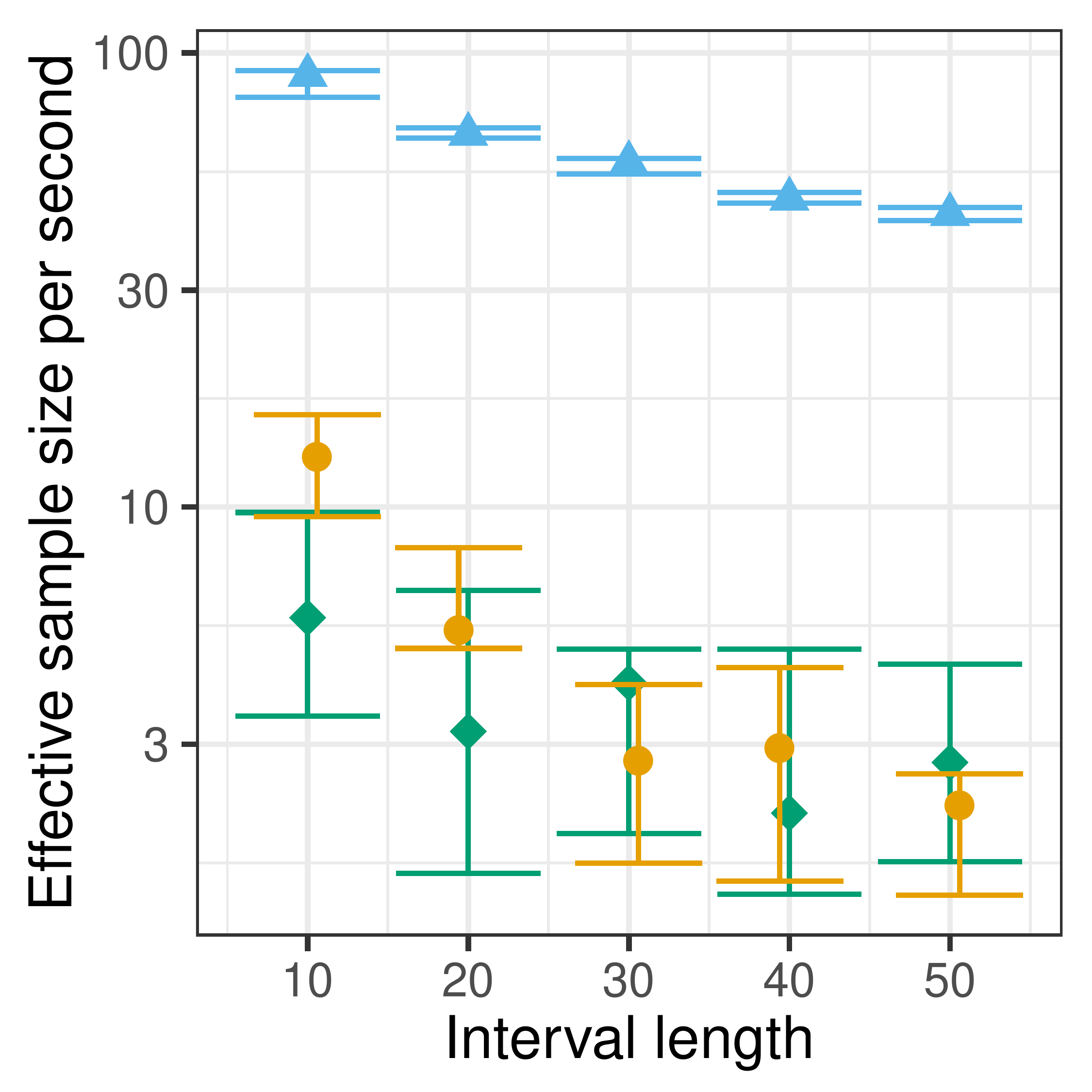}
    \vspace{-0.1 in}
  \end{minipage}
  \begin{minipage}[hp]{0.32\linewidth}
  \centering
    \vspace{-0.1 in}
    \includegraphics[width=1\textwidth]{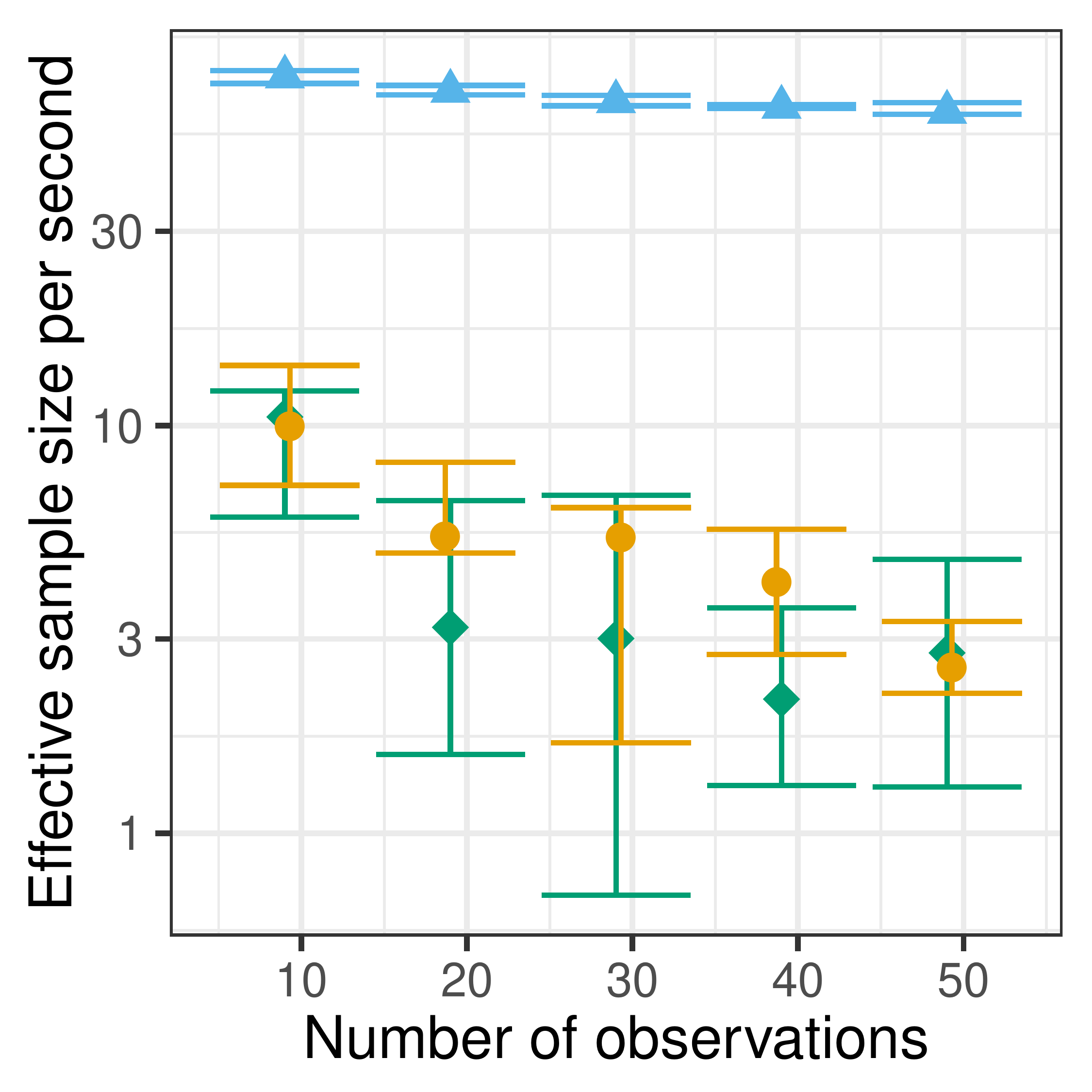}
    \vspace{-0.1 in}
  \end{minipage}
  \begin{minipage}[hp]{0.32\linewidth}
  \centering
    \vspace{-0.1 in}
  \includegraphics[width=1\textwidth]{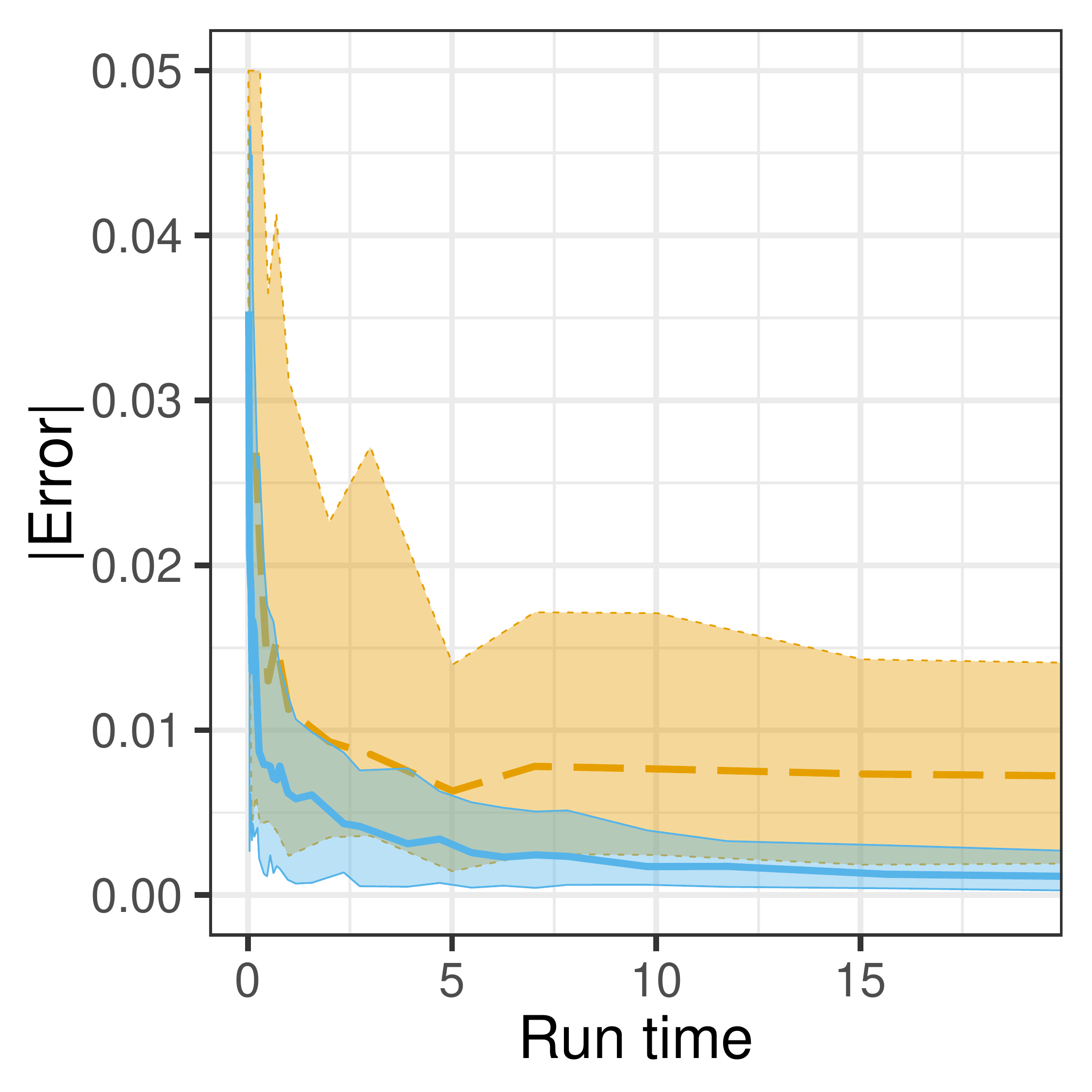}
    \vspace{-0.1 in}
  \end{minipage}
    \vspace{-0.1 in}
    \caption{ESS/s of our Gibbs sampler {\footnotesize $\blacktriangle$} and $50$-particle \fearn\  $\bullet$ and \eul\ {\footnotesize $\blacklozenge$} for (left) increasing $T$ with $N=20$ observations, and (middle) increasing $N$ with $T=20$. The right panel shows error in filtering distributions of our method (solid line) and particle filtering \pf\ (dashed line) for increasing computational budget.}
  \label{fig:hmc1_pmcmc1_nObs_tEnd}
    \vspace{-0.1 in}
\end{figure}

\begin{figure}[H]
  \begin{minipage}[hp]{0.32\linewidth}
  \centering
  \includegraphics[width=\textwidth]{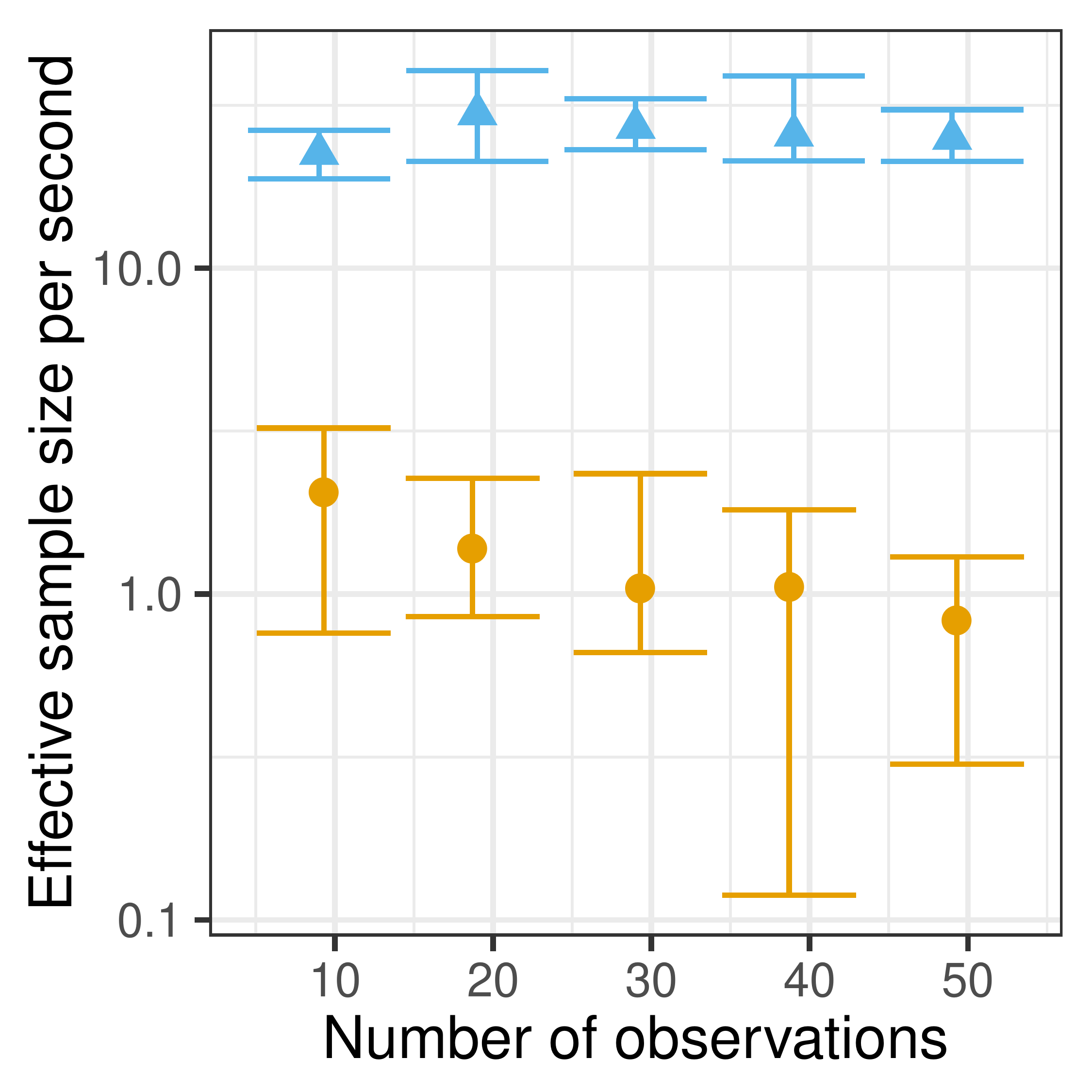}
  \end{minipage}
  \begin{minipage}[hp]{0.32\linewidth}
  \includegraphics[width=\textwidth]{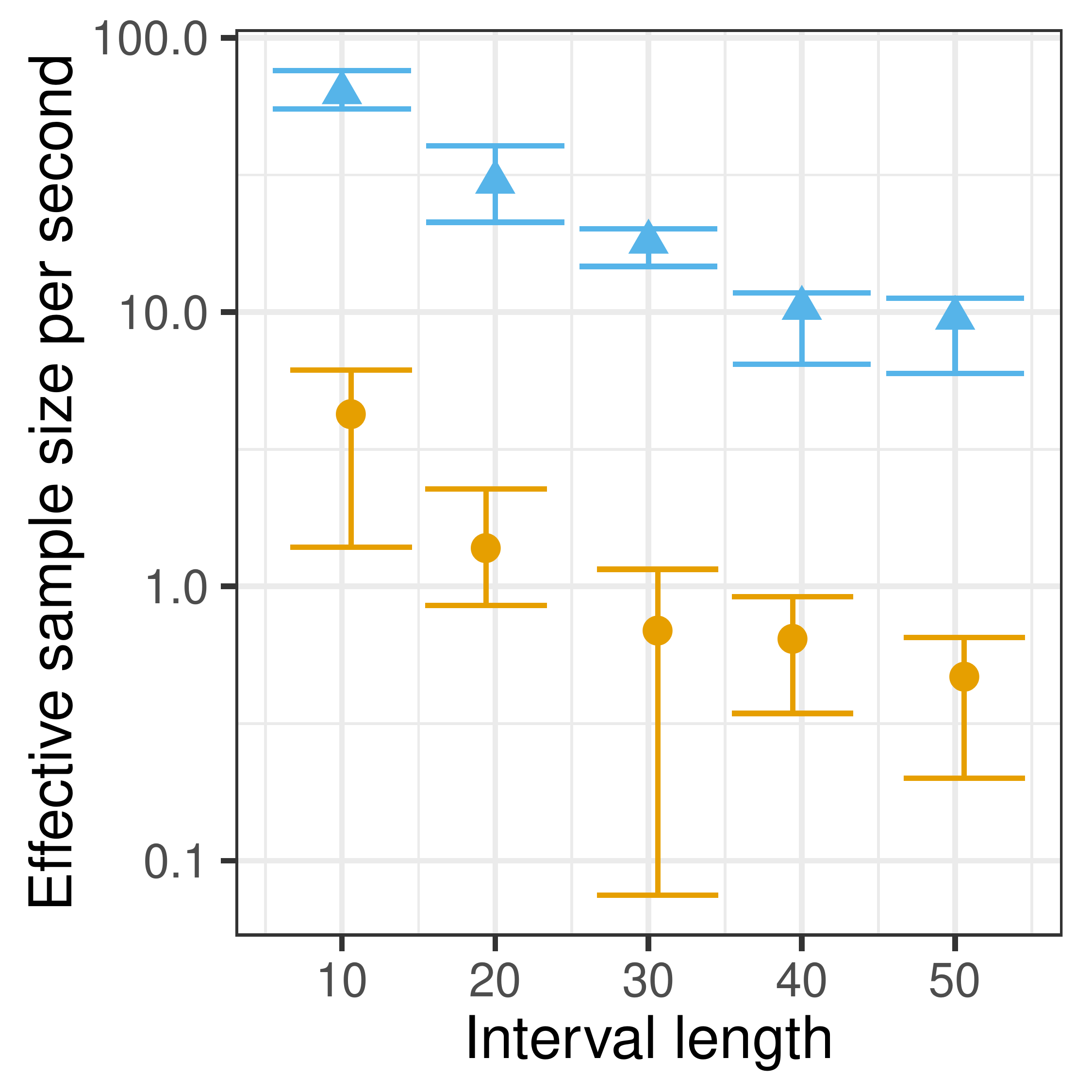}
  \end{minipage}
  \begin{minipage}[hp]{0.32\linewidth}
  \caption{ESS/s for the sine diffusion as we increase the number of observations (left) and the interval-length $T$ (right). {\footnotesize $\blacktriangle$} represents our Gibbs sampler, and $\bullet$ represents a particle-independent MH sampler based on an Euler–Maruyama discretization with stepsize 0.01.  }
  \label{fig:hmc_sin_pars}
  \end{minipage}
\end{figure}


The earlier experiments fixed the parameter $\theta$ to 0, in figure \ref{fig:hmc_sin_pars} we place of uniform $\text{Unif}(-\pi,\pi)$ on $\theta$, and compare performance of our sampler and a PIMH sampler based on \eul. 
Once again, our sampler is about an order of magnitude more efficient.

    \vspace{-0.2 in}
\section{Modeling stock prices}
    \vspace{-0.1 in}
\label{sec:example}
In our final experiment, we consider a real dataset of stock prices of Alphabet Inc.\footnote{Obtained from 
  \url{https://finance.yahoo.com/quote/GOOG?p=GOOG&.tsrc=fin-srch}}, selecting one observation each week from April 2013 to Aug 2017. 
The resulting dataset consists of $179$ observations, the first $146$ of which we used as training, and the last $33$ as test. 
We plot the data in the left panel of Figure~\ref{fig:stock_price}.  
As is typical, we preprocess the data, removing the linear trend, and then taking the logarithm of the detrended stock price. 
We also rescale time between $0$ and $10$.
Write $S_t$ for the transformed measurement at time $t$. 
For $n$  trading days $O = \{o_1, o_2, \dotsc, o_n\}$, our observations are $S =\{s_{o_1}, s_{o_2}, \dotsc, s_{o_n}\}$. 
The right panel in Figure~\ref{fig:stock_price} plots this transformed data.
\begin{figure}[H]
  \centering
  \begin{minipage}[hp]{0.35\linewidth}
  \centering
    \vspace{-0.1 in}
    \includegraphics[width=1\textwidth]{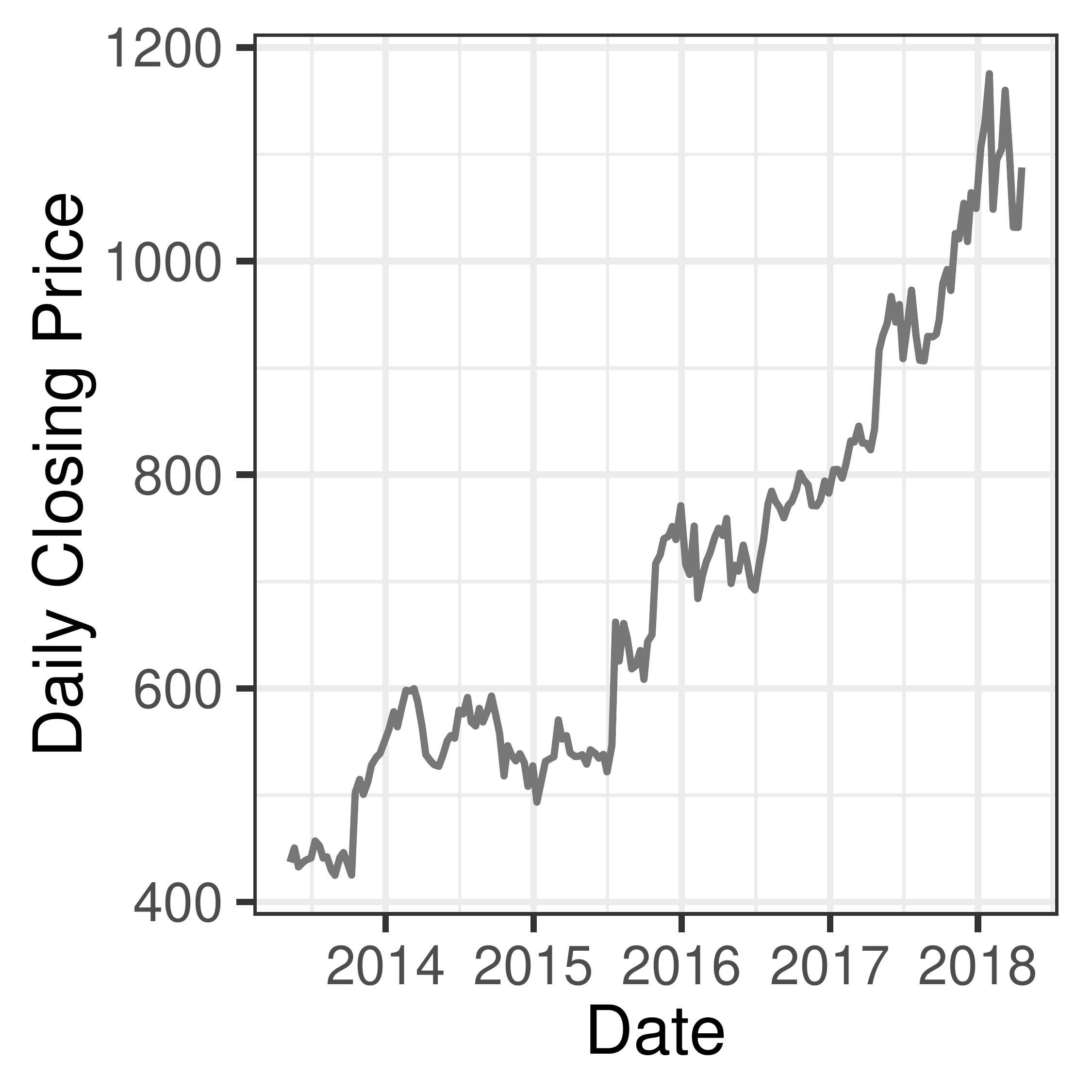}
    \vspace{-0.1 in}
  \end{minipage}
  \begin{minipage}[hp]{0.35\linewidth}
  \centering
    \vspace{-0.1 in}
    \includegraphics[width=1\textwidth]{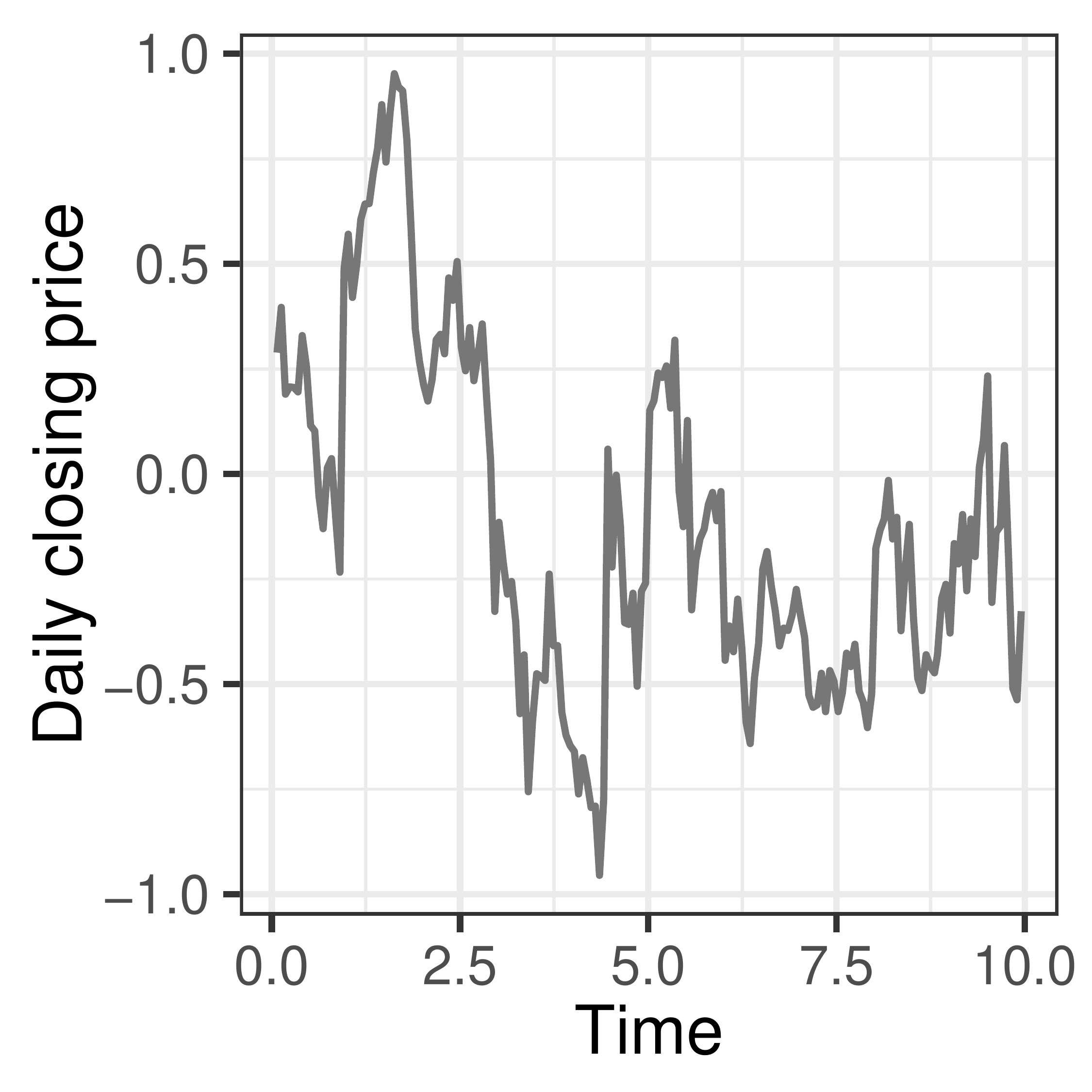}
    \vspace{-0.1 in}
  \end{minipage}
  \begin{minipage}[hp]{0.28\linewidth}
    \vspace{-0.8 in}
    \caption{Weekly stock prices for Alphabet Inc, from April 2013 to April 2018. The left panel shows the raw data, and the right one shows the transformed data which we model.}
  \label{fig:stock_price}
    \vspace{-0.1 in}
  \end{minipage}
    \vspace{-0.4 in}
\end{figure}

While stock prices have classically been modeled by geometric Brownian motion~\citep{black1973pricing}, limitations of such models, such as their inability to capture empirically observed heavy tails, have been well documented. 
\citet{bibby1996hyperbolic} used a hyperbolic distribution to model the increments of the process, and we use the hyperbolic diffusion of equation~\eqref{eq:hyper1}. 
We treat this as a latent process underlying the observed stock prices $S$. 
The observations themselves are modeled as additive Gaussian perturbations of the underlying diffusion. 
The overall model is 
\begin{align}
  \W_0 & \sim \prior, \qquad \qquad \df \W_t = - \frac{\theta\W_t}{\sqrt{1+\W_t^2}} \df t + \df B_t, \qquad
  s_t \sim \mathcal{N}(\W_t, \sigma^2),\quad t \in \{o_1,\dotsc,o_n\}.
\end{align}
For simplicity, we fix the standard deviation of the measurement noise to $0.2$, though we could easily place a conjugate prior on this.
We consider two settings, one with $\theta$ fixed to 1, and the second with a rate-1 exponential prior over $\theta$.
We apply our MCMC algorithm to the data in both settings. 
Figure~\ref{fig:hyper_trace_acf1} shows the MCMC traceplot and autocorrelation function of the trajectory value at $\W_{T/2}$, the midpoint of the interval, for the sampler that updates $\theta$. 
The results with $\theta$ fixed are similar: our sampler mixes well, with no significant autocorrelation at lags larger than 5. 
\begin{figure}[]
  \centering
  \begin{minipage}[hp]{0.34\linewidth}
  \includegraphics[width=.99\textwidth]{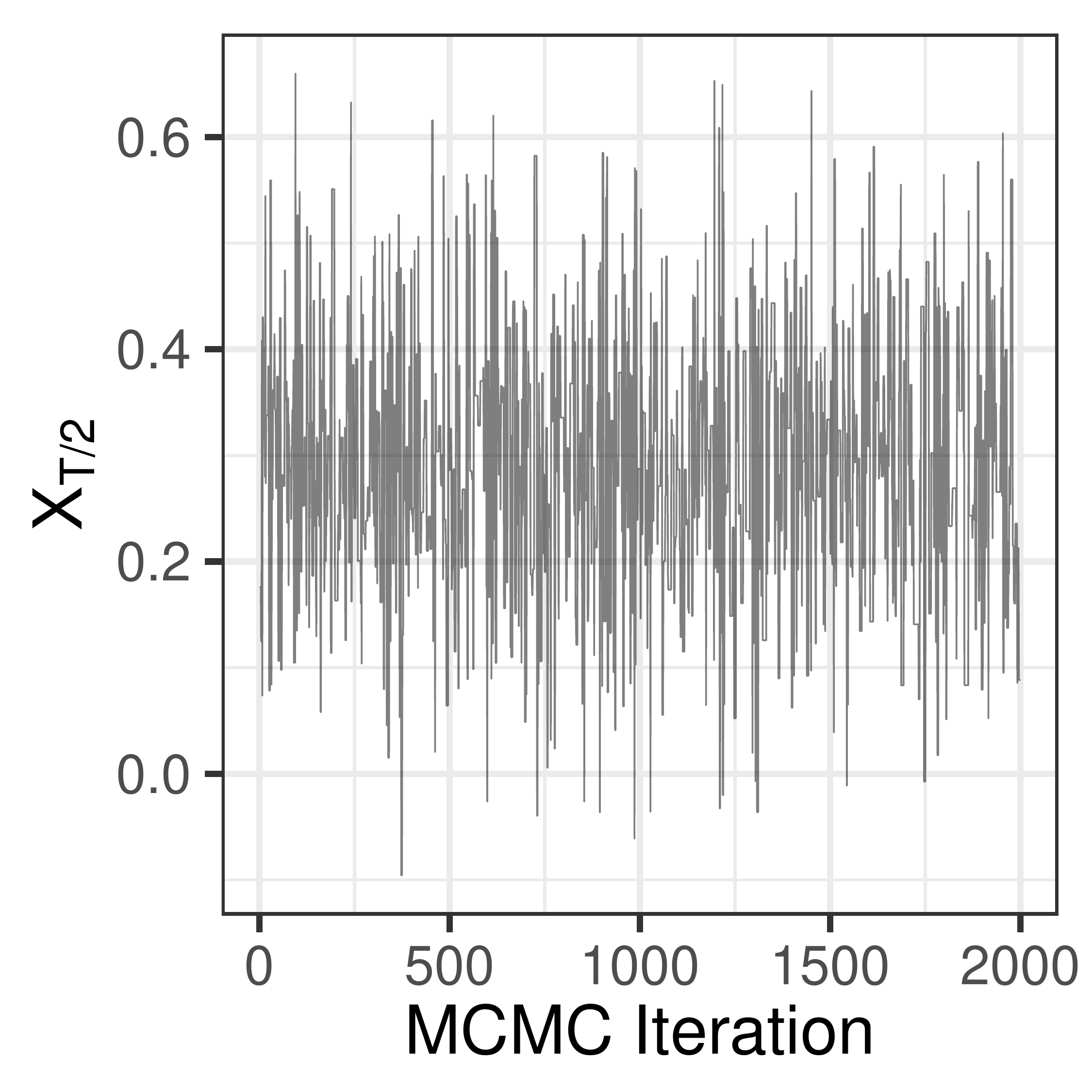}
\end{minipage}
  \begin{minipage}[hp]{0.34\linewidth}
  \includegraphics[width=.99\textwidth]{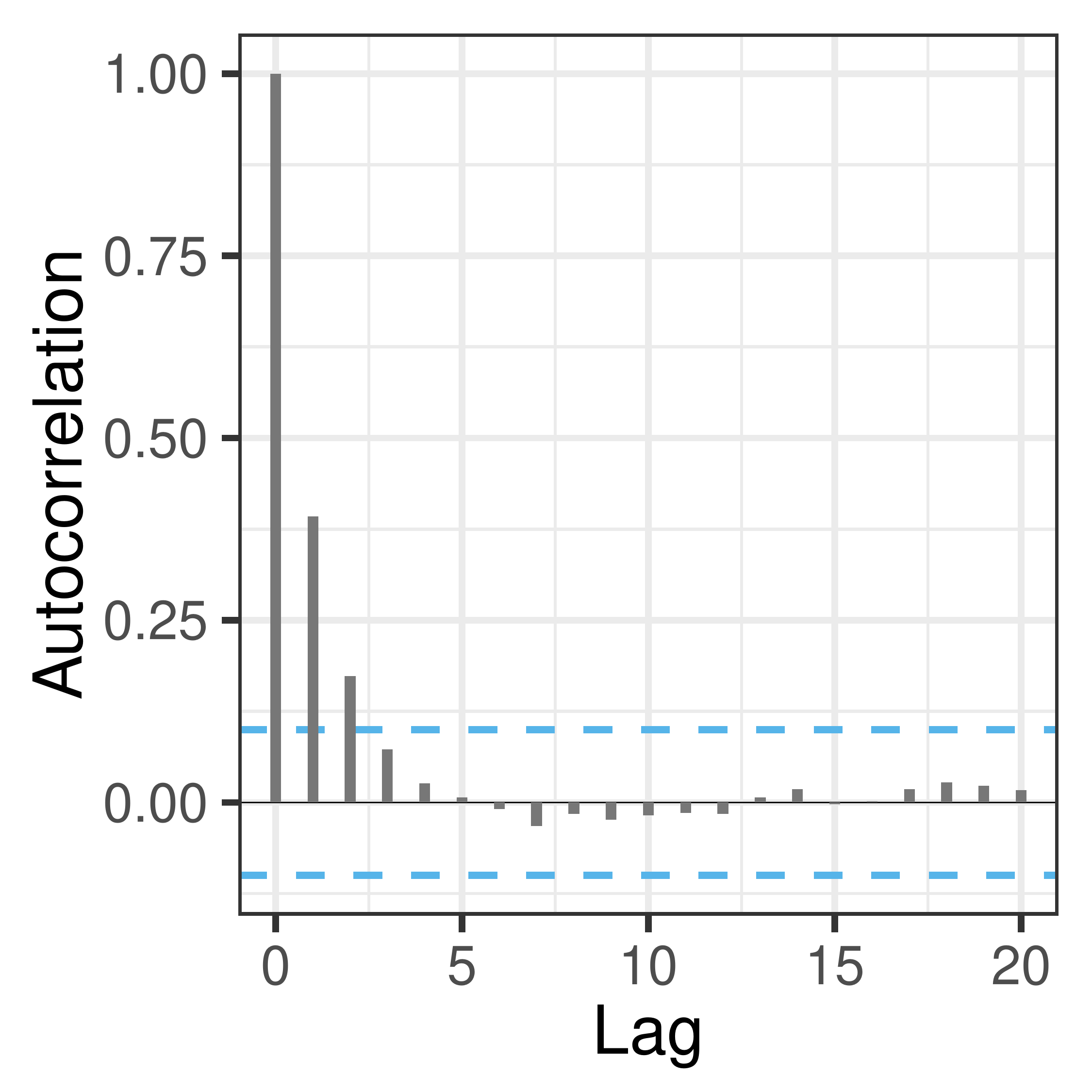}
\end{minipage}
  \begin{minipage}[hp]{0.3\linewidth}
  \caption{(left) MCMC trace of $\W_{T/2}$, the diffusion value at the midpoint of the observation interval, and (right) its autocorrelation function.}
  \label{fig:hyper_trace_acf1}
\end{minipage}
  \centering
  \begin{minipage}[hp]{0.32\linewidth}
  \centering
    \includegraphics[width=1\textwidth]{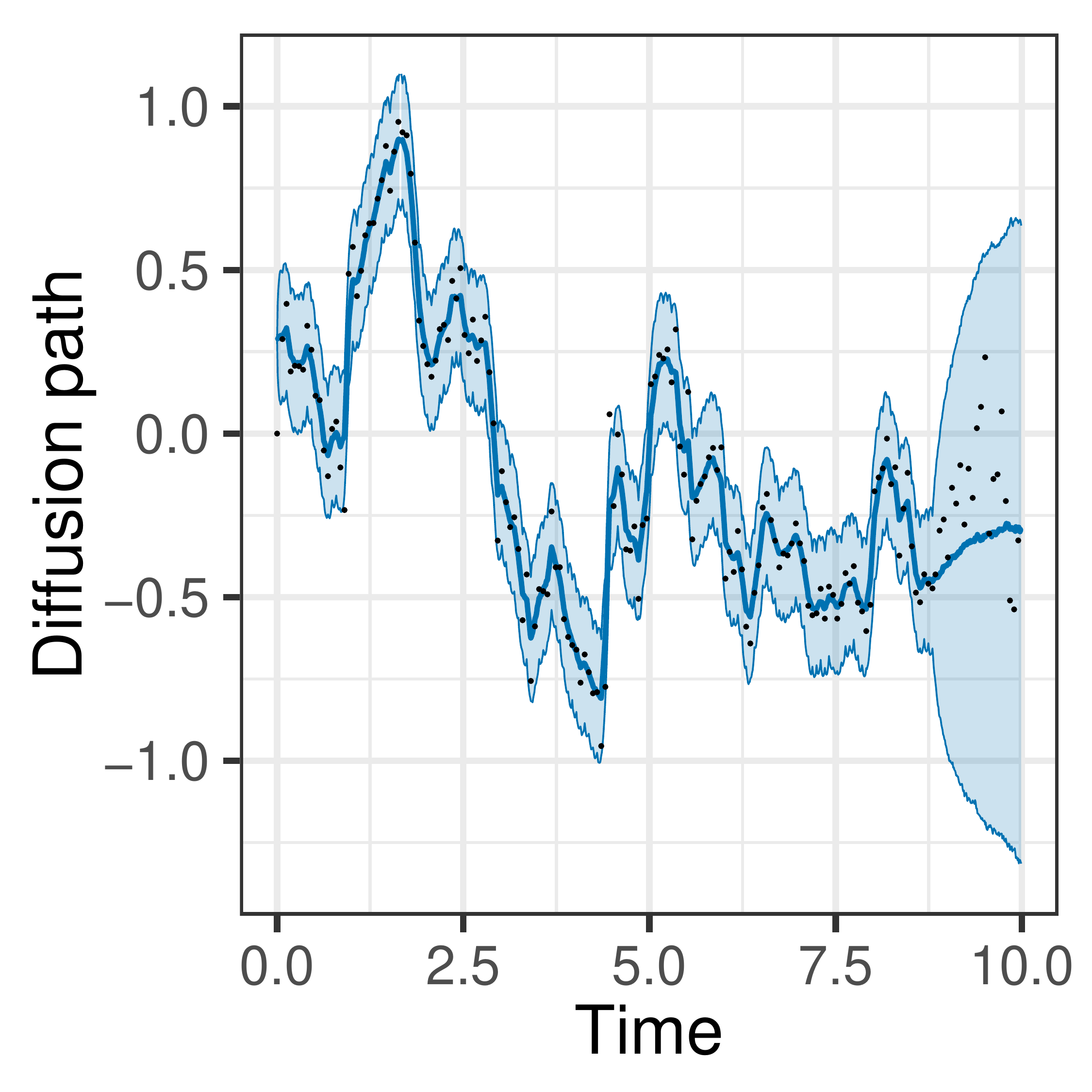}
    \vspace{-0.1 in}
  \end{minipage}
  \begin{minipage}[hp]{0.32\linewidth}
  \centering
    \vspace{-0.1 in}
    \includegraphics[width=1\textwidth]{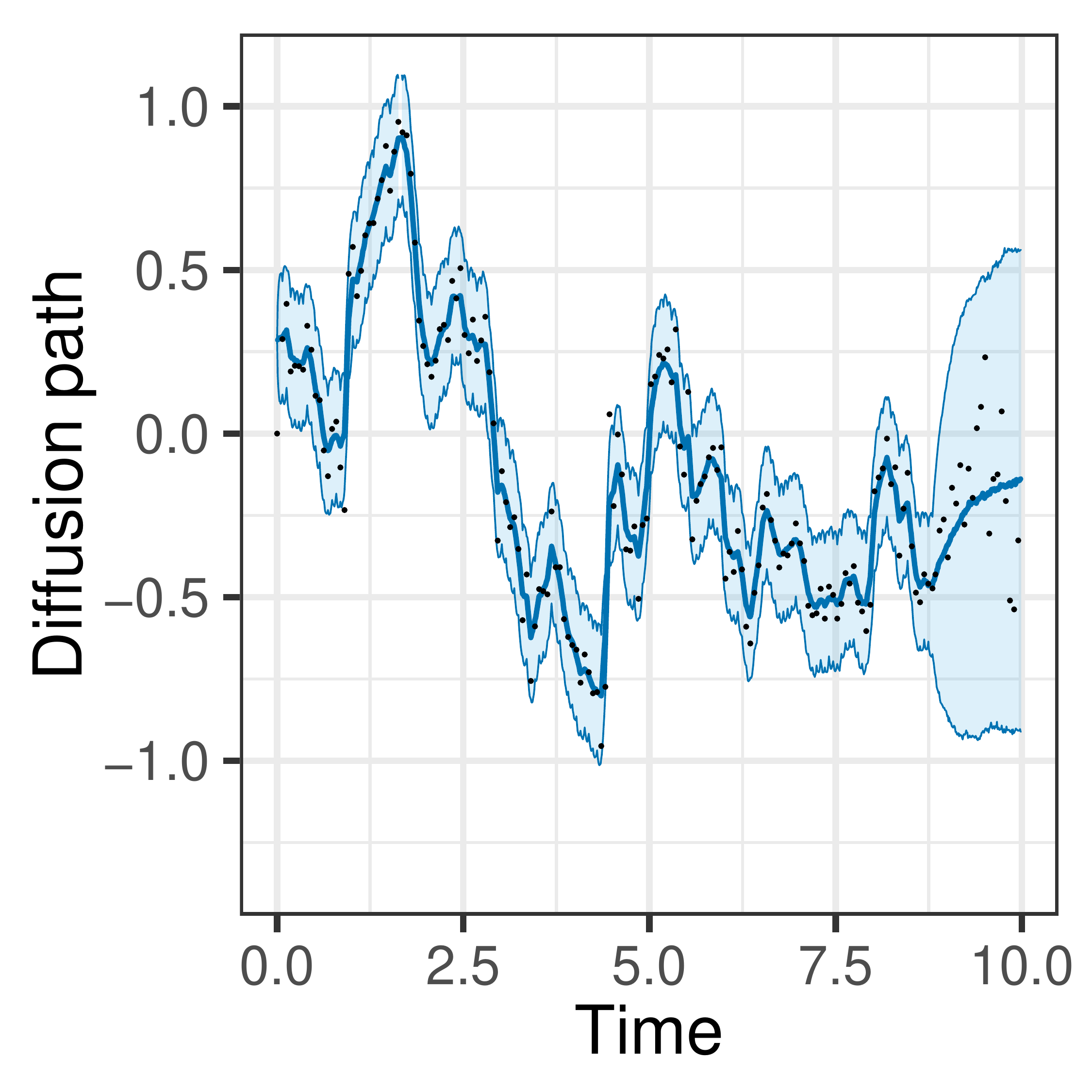}
    \vspace{-0.1 in}
  \end{minipage}
  \begin{minipage}[hp]{0.32\linewidth}
  \centering
    \vspace{-0.1 in}
    \includegraphics[width=1\textwidth]{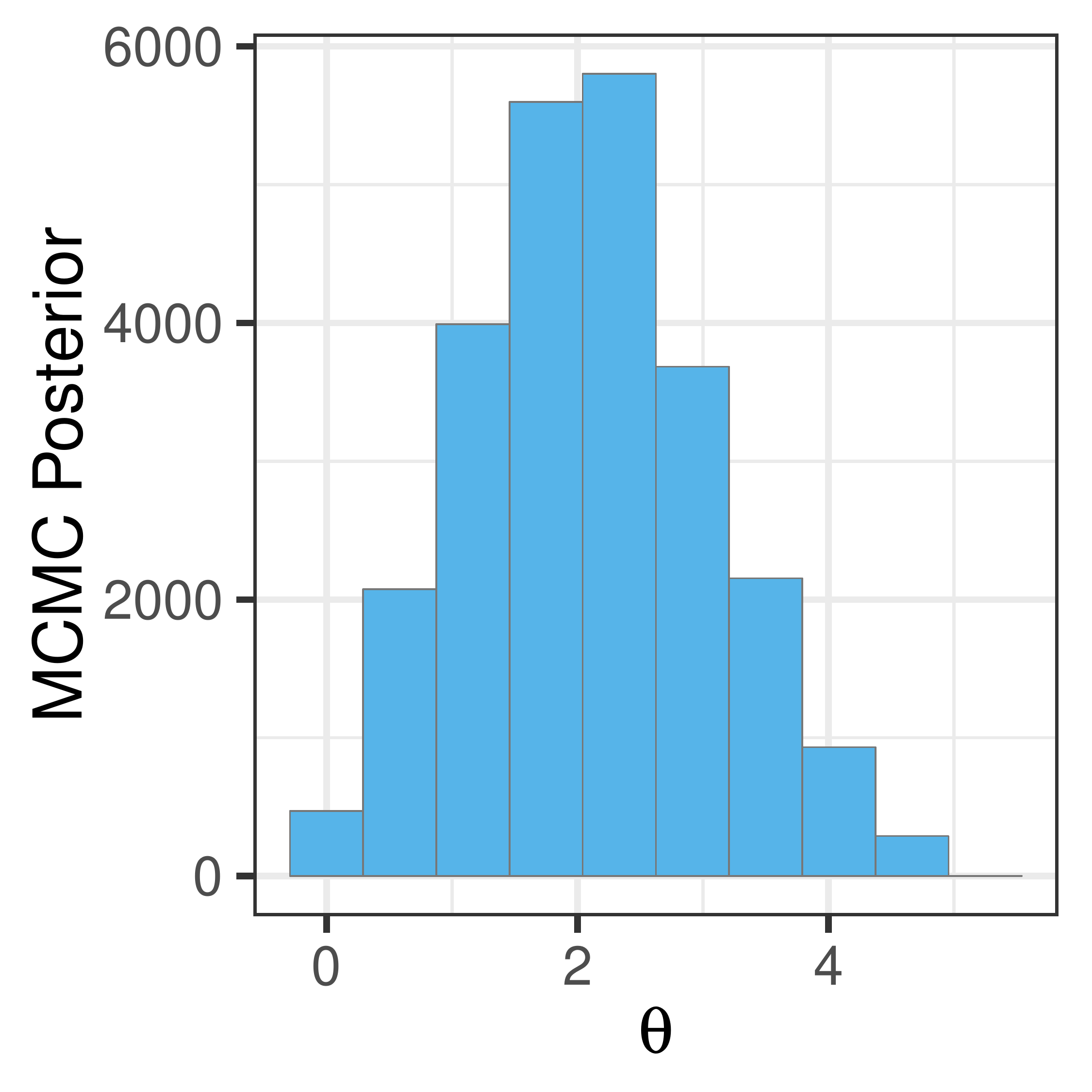}
    \vspace{-0.1 in}
  \end{minipage}
    \vspace{-0.2 in}
  \caption{Median and 90\% quantiles of the posterior over paths with (left) $\theta=1$, and (center) $\theta$ updated after placing a rate-1 exponential prior. (Right) is the corresponding posterior over $\theta$.}
  \label{fig:example_pmcmc2}
  \centering
    \vspace{0.15 in}
  \begin{minipage}[hp]{0.35\linewidth}
  \centering
    \includegraphics[width=1\textwidth]{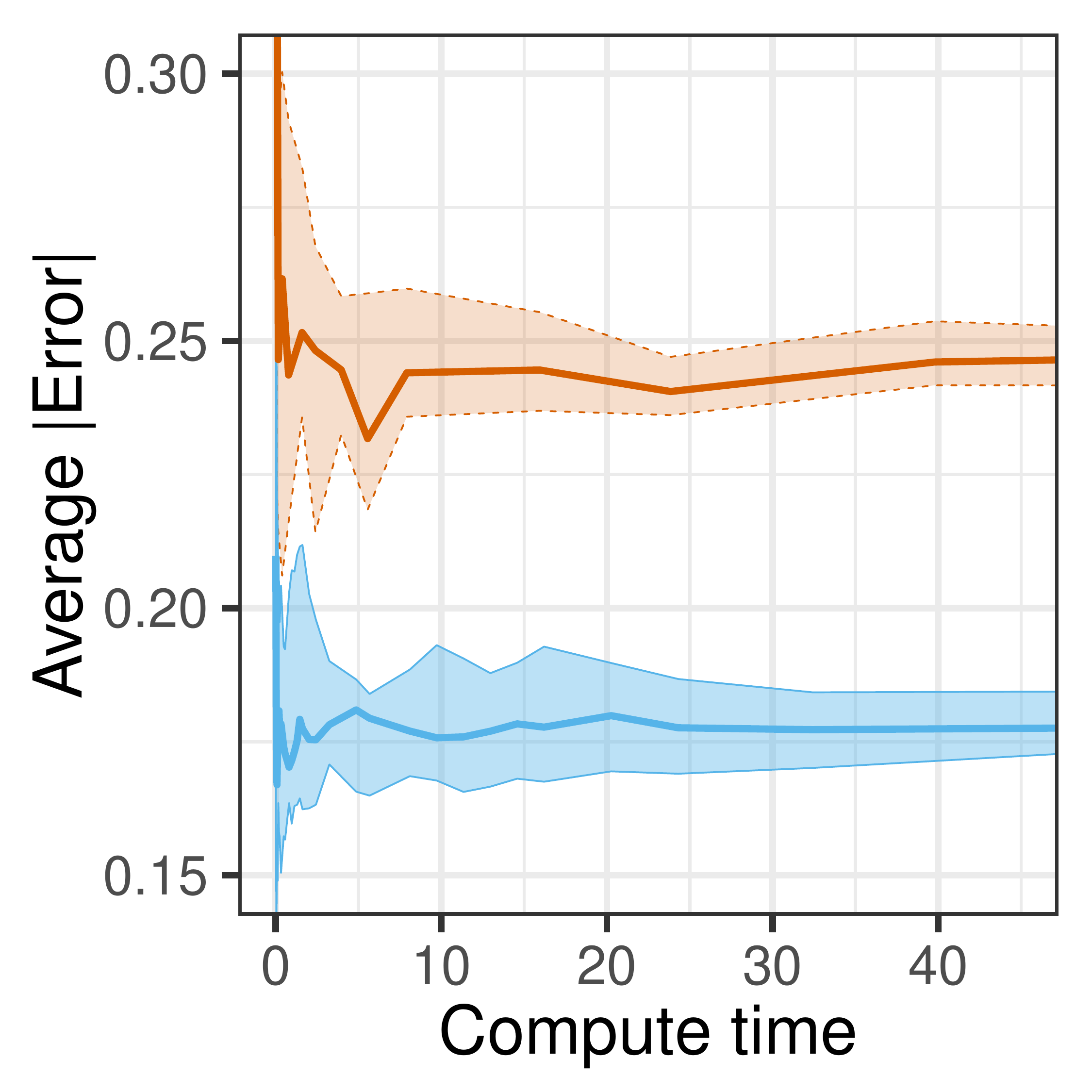}
    \vspace{-0.1 in}
  \end{minipage}
  \begin{minipage}[hp]{0.35\linewidth}
  \centering
    \vspace{-0.1 in}
    \includegraphics[width=1\textwidth]{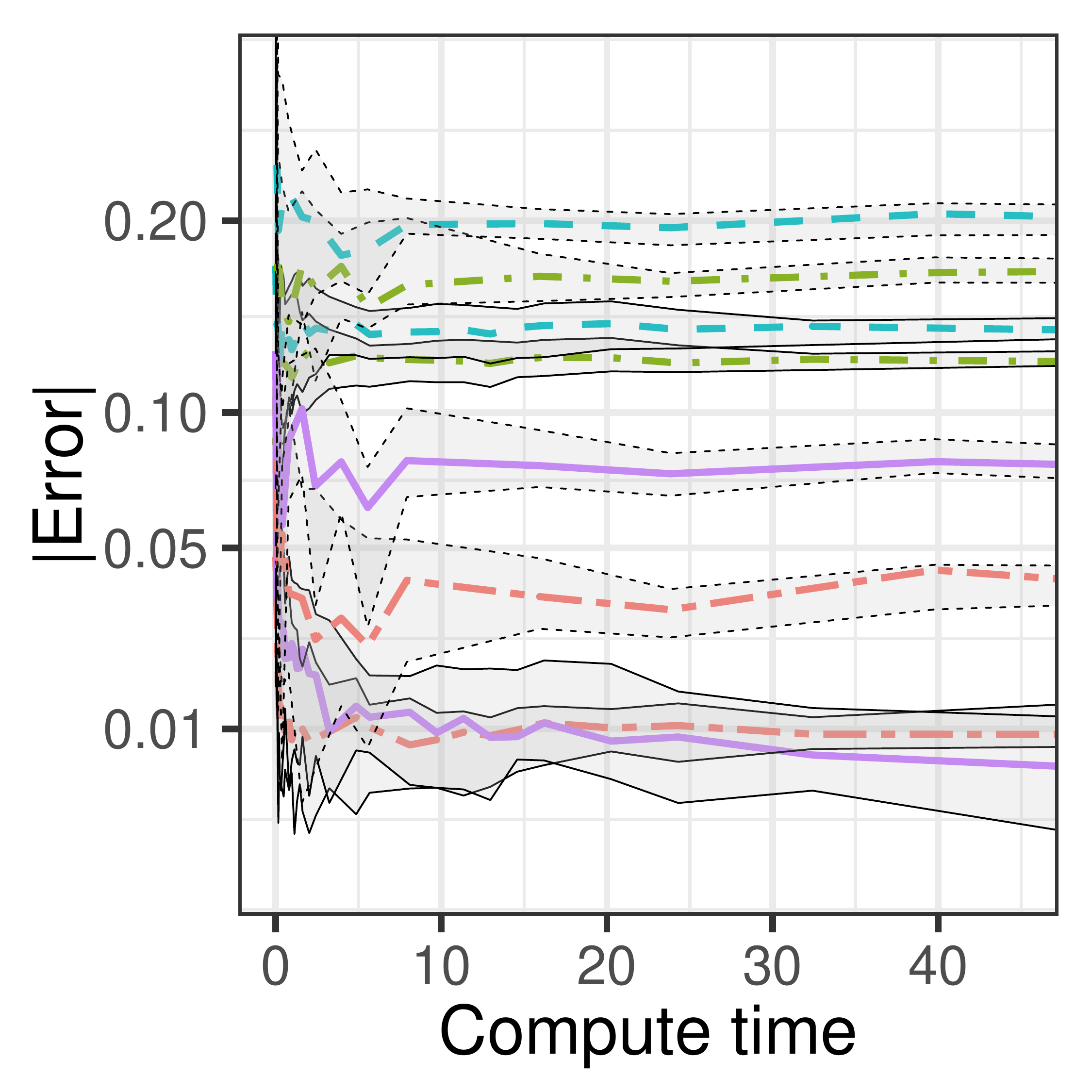}
    \vspace{-0.1 in}
  \end{minipage}
  \begin{minipage}[hp]{0.28\linewidth}
    \vspace{-0.4 in}
    \caption{ Absolute error (left) averaged across all test datapoints, and (right) for first 4 test datapoints (each a different line-type). Solid/dashed ribbons are our method/\pf. }
    \label{fig:stock_pred}
  \end{minipage}
\end{figure}

The left and middle panels in figure~\ref{fig:example_pmcmc2} plot the posterior distribution over diffusion paths, showing the median and a $90\%$ posterior credible interval.
The left is with $\theta$ fixed to $1$ and the middle is when we update $\theta$.
We see that the posterior spread includes most observations, suggesting that our model, viz.\ the hyperbolic diffusion with Gaussian noise produces a good fit for this dataset. 
We note that we only fit the models on training data (until about time $8$), the remaining datapoints are held-out test data that the models never see.
This explains the increase in uncertainty towards the end of the interval, though both models cover the observations.
Updating $\theta$ provides a tighter fit by virtue of infering stronger zero-reverting dynamics.
This is reflected in the posterior over $\theta$ (the rightmost panel) which concentrates on values larger than $1$.

Our algorithm is significantly faster that the two pMCMC baselines \eul\ and \fearn, and we do not include speed-accuracy comparisons with these. 
We note though that posterior approximations produced by all 3 MCMC algorithms agreed with each other, with 2-sample Kolmogorov-Smirnov tests failing to reject the null that they come from the same distribution.
Instead, we compare our algorithm with the random-weight particle filtering algorithm \pf.
We set $\theta$ to $1$ for both algorithms, and ran them on the training dataset. 
We then imputed path values on the test times, and calculated the absolute difference of the posterior predictive means at test times from the observed test values.
We repeated this 10 times for each algorithm for different random seeds, plotting the errors for increasing computational budgets.
As before, increasing the computational budget of our method involved running for more MCMC iterations, while for particle filtering, this involved rerunning the algorithm with more particles.
The left panel in figure~\ref{fig:stock_pred} shows the average absolute error across all 33 test datapoints: the ribbons are 90\% quantiles across multiple initializations, and the thick line is the median. 
Our algorithm (solid ribbons) outperformed \pf\ (dashed ribbons), producing more accurate results for comparable runtimes. 
The right limit of the x-axis (at 45 seconds) corresponds to our method run for about 20000 samples, and \pf\ run with 5000 particles. 
The right panel disaggregates these results, showing the absolute error for the first 4 test datapoints. These are typical of the results we observed: our algorithm outperformed \pf\ on most test data points.

\vspace{-.2in}
\section{Discussion}
\label{sec:conc}
\vspace{-.1in}
In this paper, we proposed a computationally efficient auxiliary variable Gibbs sampling algorithm that allows simulation from the EA1 class of SDEs without any discretization error.
Our sampler builds on the EA1 rejection sampling algorithm for diffusions, described in~\citet{beskos2005exact} and follow-up work.  
Our method allows prior simulation from the SDE, conditional simulation given noisy observations as well as parameter inference.

There are a number of avenues for future research. 
In follow-up work, Beskos and collaborators developed exact rejection sampling algorithms for larger classes of SDEs. 
Recall that the EA1 class is limited to SDEs where $\phi(\cdot) = \alpha^2(\cdot) + \alpha'(\cdot) \in [L, M+L]$ where $L$ and $M$ are finite. 
In~\citet{beskos2006retrospective, beskos2006exact}, the authors also consider two broader classes, EA2 where $\phi(\cdot)$ is bounded only from one side, and EA3, where it is not bounded at all.
Extending our MCMC scheme to such situations is conceptually straightforward, though much more involved: we need to augment our MCMC state-space to include the maximum and/or minimum of the trajectory.
Conditioned on these, imputing the SDE trajectory will involve simulating from Bessel processes instead of Brownian bridges.
We are currently exploring efficient ways to do this.
Work in~\cite{giesecke2013exact, pollock2016} has extended ideas from~\citet{beskos2005exact} to other stochastic processes, such as jump-diffusions processes, and similar ideas to this paper can be applied in that context, and to other diffusions not absolutely continuous with respect to Brownian motion.
Finally, it is interesting to better understand theoretically the convergence properties of our proposed MCMC algorithm.

\bibliography{refvr}

\newpage
\section{Appendix}
\begin{algorithm}[H]
  \caption{Euler-Maruyama algorithm~\citep{kloeden2012numerical} to simulate a diffusion process }
   \label{alg:Euler}
  \begin{tabular}{l l}
    \textbf{Input:  } & \text{A regular grid $G = \{0, t_1, t_2, \cdots, t_{n-1}, T\}$ on a time interval $[0, T]$}. \\
                     & \text{An initial distribution over states $\pi$, a drift term $\alpha(\cdot)$ and a diffusion term $\beta(\cdot)$}.\\
    \textbf{Output:  }& \text{A diffusion trajectory $\{\W_0, \W_{t_1}, \W_{t_2}, \cdots, \W_T\}$ evaluated on $G$}.\\
   \hline
   \end{tabular}
   \begin{algorithmic}[1]
      \State Simulate $\W_0 \sim \pi$
      \For{$i$ in $1$ to $N$}
      \State \text{Simulate $y_i$ from the standard normal distribution.}
      \State Set $\W_{t_{i+1}} \gets \W_{t_i} + \alpha(\W_{t_i}) (t_{i+1} - t_i) + \beta(\W_{t_i})\sqrt{t_{i+1} - t_i}y_i$
      \EndFor
   \end{algorithmic}
\end{algorithm}

\subsection{A particle MCMC algorithm for path inference}
We first describe a particle filtering algorithm to propose a new path 
$\W^*$
\begin{algorithm}[H]
  \caption{Particle filtering algorithm to simulate a diffusion process }
   \label{alg:partMCMC}
  \begin{tabular}{l l}
   \textbf{Input:  } & \text{A regular grid $G = \{0, t_1, t_2, \cdots, T\}$ on a time interval $[0, T]$}, \\
                     & \text{An initial distribution over states $\pi$, a drift term $\alpha(\cdot)$},\\
                     & \text{Observations at times $O = \{o_1,\dotsc,o_{|O|}\}$, with observation $i$ having likelihood $\ell_i(\W_{o_i})$ } \\
   \textbf{Output:  }& \text{A new trajectory $\W^*_G$ from the SDE conditioned on the observations}.\\
   \hline
   \end{tabular}
   \begin{algorithmic}[1]
      \State Sample initial states for N particles $X^k(0)$ from $\pi$, $k = 1,...,N$. 
      \For{$i$ in $1$ to $|O|$}
	\State  For $k = 1,2,...,N$, update particle $k$ from $[0,o_{i-1}]$ to $[0,o^i]$ by forward simulating 
    via the Euler-Maruyama algorithm on the grid.
    \State  Calculate the weights $w^k_i = \ell_i(X^k_{o_i})$ and normalize $W^k_i = \frac{w^k_i}{\sum_{k = 1}^N w^k_i},\ \  k = 1,2,...,N.$ 
	\State  Sample $J_{i}^k \sim \text{Multi}(\cdot| (W^1_{i},\dotsc,W^N_{i}))$ ,$k = 1,2,...,N$.
	\State  Set $X_{[0, o_i]}^k := X_{[0,o_i]}^{J^k_i},\ \  k = 1,2,...,N.$.
      \EndFor
   \end{algorithmic}
\end{algorithm}
Assume no observations at the end-time $T$. Then uniformly pick one of the $N$ particles, call this $X^*$. We have an estimate of $P_\ell(X^*)$, the conditional probability of $X^*$ given the observations:
$ P_\ell(X^*) = \prod_{i = 1}^n \left[ \sum_{k = 1}^N  \frac{1}{N} w_i^k \right].$

\subsubsection{Particle MCMC algorithm for diffusions}

\begin{algorithm}[H]
  \caption{The particle MCMC algorithm for SDE trajectories}
   \label{alg:SMC}
  \begin{tabular}{l l}
   \textbf{Input:  } & \text{A regular grid $G = \{0, t_1, t_2, \cdots, T\}$ on a time interval $[0, T]$}, \\
                     & \text{An initial distribution over states $\pi$, a drift term $\alpha(\cdot)$},\\
                     & \text{Observations at times $O = \{o_1,\dotsc,o_{|O|}\}$, with observation $i$ having likelihood $\ell_i(\W_{o_i})$ } \\
                     & \text{Current trajectory $\W_G$, parameter $\theta$, and
                     current estimate of probability $P(\W_G|O)$}.\\
   \textbf{Output:  }& \text{A new trajectory $\W^*_G$ from the SDE, new parameter $\theta^*$ and 
     new estimate $P(X^*_G|O)$}.\\
   \hline
   \end{tabular}
   \begin{algorithmic}[1]
     \State Propose a parameter $\theta^*$ from some distribution $q(\theta^*|\theta)$ 
     \State Run the particle filtering algorithm to generate a sample $X^*_G$ along with the estimate $P_\ell(X^*_G)$.\\
     Accept $(\theta^*, \W^*_G)$ with probability 
     $ \mathtt{acc} = 1 \wedge \frac{P_\ell(X^*_G) p(\theta^*) q(\theta|\theta^*)}{P_\ell(X_G )p(\theta) q(\theta^*|\theta)}.$

\end{algorithmic}
\end{algorithm}

\subsection{Details of Hamiltonian Monte Carlo updates }
The Hamiltonian  Monte  Carlo~\citep{duane1987hybrid,neal2011mcmc} sampling algorithm defines a Hamiltonian function, using the target distribution as the potential energy term, and introducing a kinetic energy term parameterized by a set of auxiliary momentum variables. 
The algorithm proceeds by updating the variables of interest (`position') as well as the momentum variables according to the Hamiltonian dynamics, keeping the Hamiltonian approximately constant.
In particular, if we want to sample from a distribution $L(q)$,
first, define $U(q) = -\log(L(q))$ to be the potential energy of position $q$. 
Then introduce an auxiliary variable called $p$ of the same dimension as $p$ and define $K(p) = \frac{1}{2} p^T M^{-1} p$ to be the kinetic energy. 
Here M is a symmetric, positive-definite mass matrix, which is typically diagonal, and is often a scalar multiple of the identity matrix.
The Hamiltonian is then defined as 
$H(q, p) = U(q) + K(p)$
In our settings, the variables of interest are the SDE path evaluated on the Poisson grid $\Psi$, as well as the observation times $O$: $q \equiv X_{\Psi\cup O}$.
The distribution of interest is given in equation~\eqref{eq:hmc_target}, and we repeat it below:
\begin{align}
  L(q) \equiv p(\W_{\Psi \cup O}) &\propto \prior(\W_0)h_{\W_0}(\W_T) \brobri(X_{O \cup \Psi}|0,X_0,T,X_T) \ell(\W_O) \prod_{g \in \Psi} \left( 1- \frac{\phi(\W_g)}{M}\right). 
  \label{eq:hmc_target_app}
\end{align}
Recall that $\prior(\W_0)$ is the distribution over the initial value of the diffusion, $h_{\W_0}(\W_T)$ is the bias term in the $h$-biased Brownian bridge, while $\ell(\cdot)$ is the likelihood term. 
The term $\brobri(X_{O \cup \Psi}|0,X_0,T,X_T)$ gives the probability of imputing values $\W_{O\cup\Psi}$ on ${O\cup\Psi}$ under a Brownian Bridge with values $\W_0$ and $\W_T$ at times $0$ and $T$.
Writing $O\cup\Psi \equiv \{t_1, \dotsc, t_{S}\}$, we have
\begin{align}
  \brobri(\{X_{t_1},\dotsc, X_{t_S}\}|0,X_0,T,X_T) &= 
  P(\W_{t_{S}} | \W_0, \W_T)   \times P(\W_{t_{S - 1}} | \W_0, \W_{t_S}) \times \cdots \times P(\W_{t_1} | \W_0,   \W_{t_2}), \nonumber \\
  \text{where } P(\W_{t_i} | \W_0, \W_{t_j}) & \sim
  \normal\left( \frac{(t_j - t_i) \W_0 + t_i \W_{t_j} }{t_j}, \frac{(t_j - t_i) t_i}{t_j}\right) \text{ for any } t_j > t_i > 0. \label{eq:brownian_br}
\end{align}
The potential energy is the logarithm of equation~\eqref{eq:hmc_target_app}, and factors into a summation of straightforward terms.
The Brownian bridge term in particular decomposes into a sum of quadratic terms.
The gradient of equation~\eqref{eq:hmc_target_app} with respect to $\W_{O\cup\Psi}$ is thus also straightforward to calculate, allowing an easy implementation of the HMC algorithm. We refer the reader to~\citet{neal2011mcmc} for more details, which are now completely standard.

\subsection{Improved mixing via tempering}

 Figure \ref{fig:sin_hmc_ea_euler} compares of our sampler with these settings with EA1 and the Euler-Maruyama approximation. The results are similar to the previous experiment: EA1 does not scale to large $T$, while our asymptotically exact sampler performs between the crude and fine discretizations.
 \begin{figure}[]
   \centering
   \begin{minipage}[]{0.34\linewidth}
   \centering
     \vspace{-0.0 in}
     \includegraphics[width=.99\textwidth]{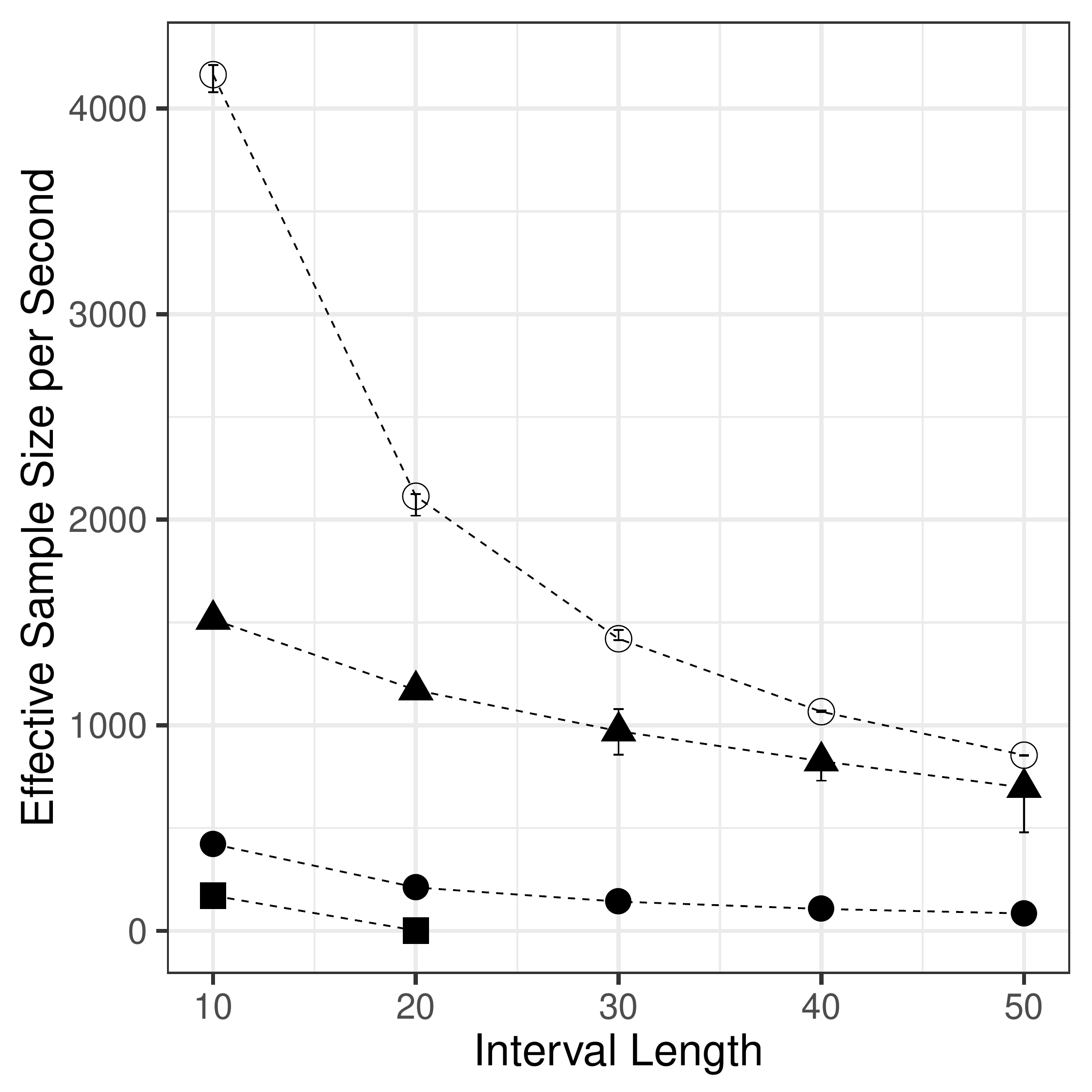}
     \vspace{-0.1 in}
   \end{minipage}
   \begin{minipage}[hp]{0.549\linewidth}
     \vspace{-1.1 in}
     \caption{ESS/s for different samplers against interval length $T$ for the SDE with a periodic drift function. {\footnotesize $\blacktriangle$} represents our method, $\bullet$ and $\circ$ represents Euler-Maruyama method with stepsize $0.01$ and $0.1$ respectively and {\tiny $\blacksquare$ } represents EA1. 
     Because of low acceptance rates, we did not run EA1 for interval lengths longer than 20.}
   \label{fig:sin_hmc_ea_euler}
     \vspace{-0.6 in}
   \end{minipage}
     \vspace{-0.4 in}
 \end{figure}

We introduce a more general and flexible {\em tempering} scheme~\citep{swendsen1986replica,neal1996sampling} to explore the trajectory space more effectively. 
For concreteness, consider the periodic diffusion in equation~\eqref{eq:sin1}.
We introduce an `inverse temperature' parameter $c$, and define a family of SDEs indexed by $c$: 
    \vspace{-0.1 in}
\begin{align}
\df \W_t = c \sin(\W_t) \df t + \df B_t, \quad c \in [0, 1]. 
\label{eq:temper1}
\end{align}
Observe that $c=0$ sets the drift term to $0$, and reduces the SDE to Brownian motion, while $c = 1$ recovers the SDE of interest. 
Intermediate values of $c$ interpolate between these two processes, with smaller values of $c$ having smaller repulsion away from $0$, and thus being easier for MCMC exploration.
It is easy to derive the EA1 sampling functions associated with an arbitrary $c$: 
\begin{align}
  A_c(u) &= c - c\cos(u), \ \  
  \phi_c(x) = \frac{c^2 \sin^2(x)}{2} + \frac{c \cos(x)}{2} + \frac{c}{2}, \ \ 
  M_c = \max(\phi_c(x)) = \frac{c^2 + c}{2} + \frac{1}{8}.
  \label{eq:temper2}
\end{align}
Our parallel tempering scheme picks a set of values for $c$, spanning the interval $[0,1]$ and including $1$. 
We focus here on using six values, $\{0,.2,.4,.6,.8,1\}$. 
Our approach is then to develop an MCMC sampler which targets a {\em joint} distribution over {six independent trajectories}, each marginally distributed according to equation~\eqref{eq:temper1} for one of the settings of $c$. 
The target distribution is thus a product distribution over the individual SDEs for each $c$. 
A simple MCMC step that targets this uses our Gibbs sampler to update each of the paths independently. 
Equation~\eqref{eq:temper2} includes the terms needed for this. 
This by itself does not solve the problem of poor mixing. However as mentioned earlier, we expect samplers corresponding to small $c$'s to explore the trajectory space better.
We exploit this to improve mixing for larger $c$'s, and thus for our SDE of interest, with $c=1$. 
In particular, we intersperse the previous trajectory-wise update steps with a {\em swap} proposal that uniformly picks two neighboring $c$'s, and proposes to exchange their associated MCMC states. 
In other words, for a chosen pair $i$ and $j$, with inverse-temperatures, $c^{(i)}$ and $c^{(j)}$, we propose swapping the associated skeletons $(\Psi^{(i)}, X^{(i)}_{\Psi^{(i)}})$ and $(\Psi^{(j)}, X^{(j)}_{\Psi^{(j)}})$. 

Write $P_{c}(\Psi, X_{\Psi})$ for the probability of the skeleton $(\Psi, X_{\Psi})$ under the measure ${}_c\QQ^{+}$ corresponding to inverse temperature $c$. 
This is just the product of equation~\eqref{eq:joint_data} with the probability of $\Psi$ under a rate $M_{c}$ Poisson process, with all terms given in equation~\eqref{eq:temper2}. 
Then the swap proposal is accepted with Metropolis-Hastings probability given by
\begin{align} 
\text{acc} = \min\left(1, \frac{P_{c_i}(\Psi^{(j)}, X^{(j)}_{\Psi^{(j)}}) \cdot P_{c_j}(\Psi^{(i)}, X^{(i)}_{\Psi^{(i)}})}{P_{c_i}(\Psi^{(i)}, X^{(i)}_{\Psi^{(i)}})\cdot  P_{c_j}(\Psi^{(j)}, X^{(j)}_{\Psi^{(j)}})}\right).
\end{align} 
Having a larger number of $c$'s will mean that the SDEs corresponding to two adjacent $c$'s will be similar, increasing the probability of acceptance. 
Of course, this comes at the price of more computation. 
Our choice of $6$ values (and thus $5$ auxiliary tempered chains) was made without too much care, and it is possible to be more systematic doing this.

\subsection{Miscellaneous results}

\begin{figure}[H]
  \centering
  \begin{minipage}[hp]{0.42\linewidth}
    \includegraphics[width=\textwidth]{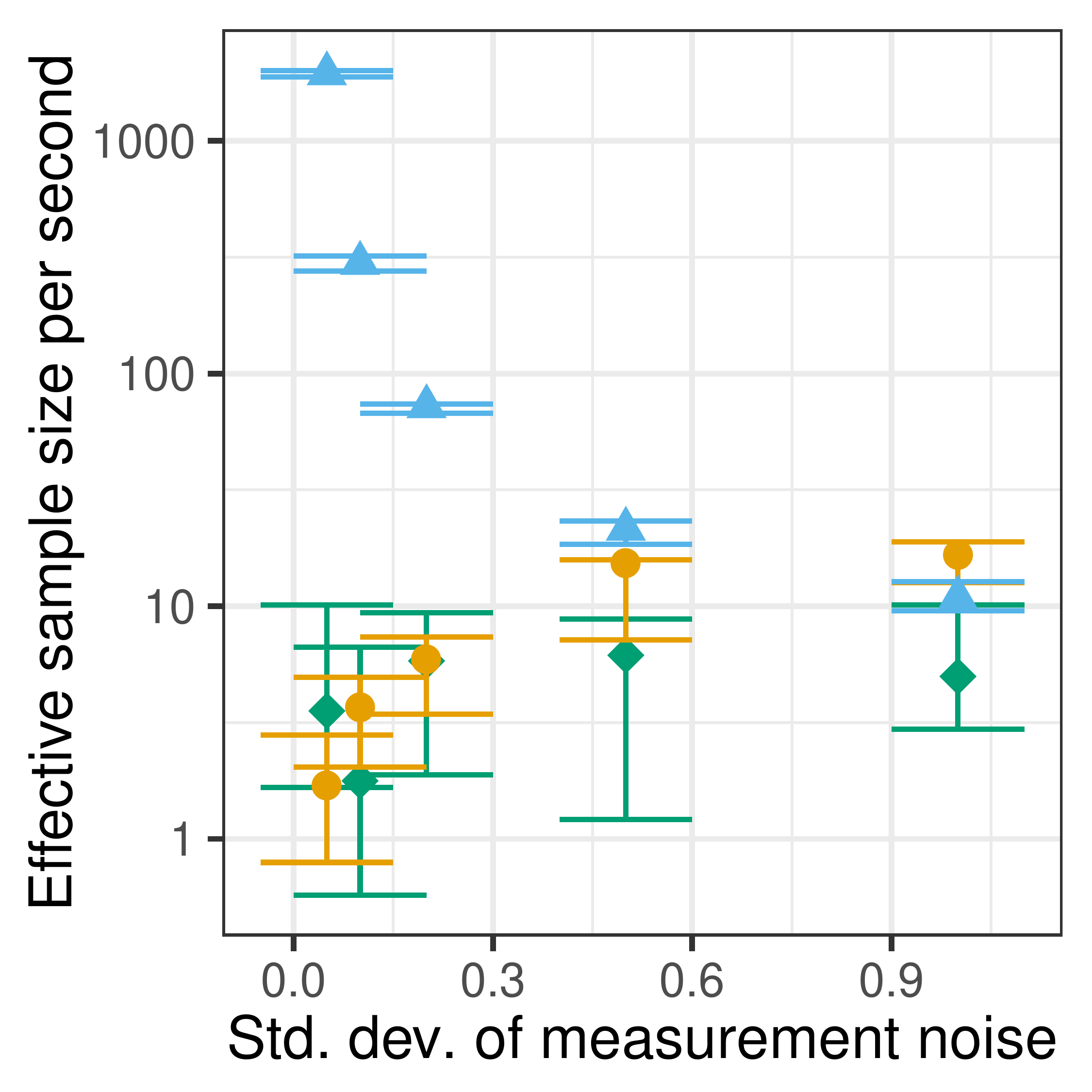}
  \end{minipage}
  \begin{minipage}[hp]{0.42\linewidth}
  \includegraphics[width=\textwidth]{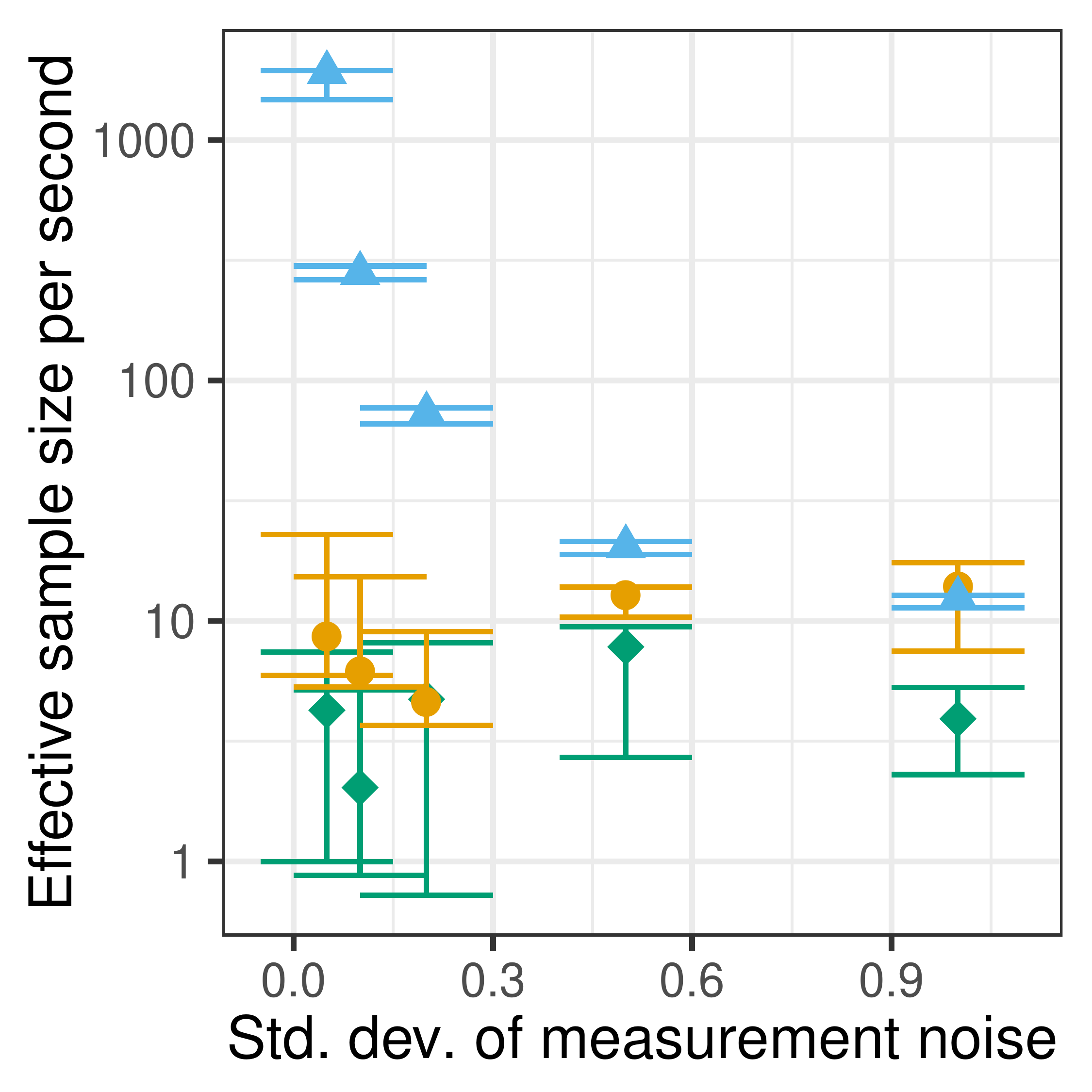}
  \end{minipage}
\caption{ESS/s of our Gibbs sampler {\footnotesize $\blacktriangle$} and $50$-particle pMCMC samplers \fearn\ $\bullet$ and \eul\ {\footnotesize $\blacklozenge$} for $T=20$ and $N=20$ as the standard deviation of the measurement noise increases. (Left) is for the hyperbolic diffusion, and (right) is for the periodic diffusion.} 
\vspace{-.2in}
\end{figure}

\end{document}